\documentclass[11pt,reqno]{article}
\pdfoutput=1

\usepackage[ascii]{inputenc}
\usepackage[bookmarksnumbered,hypertexnames=false,colorlinks=true,linkcolor=blue,urlcolor=blue,citecolor=blue,anchorcolor=green,breaklinks=true,pdfusetitle]{hyperref}
\usepackage{bbm,braket,microtype,mathrsfs,amsmath,amssymb,xcolor,amsthm,mathtools,fullpage,graphicx,enumitem,caption,booktabs,bm,xspace}
\usepackage[longend,ruled,linesnumbered,nokwfunc]{algorithm2e}

\usepackage{footnote}
\makesavenoteenv{tabular}
\makesavenoteenv{table}
\usepackage{multirow}
\usepackage[T1]{fontenc}
\usepackage[english]{babel}
\usepackage[capitalize]{cleveref}
\usepackage{thm-restate}
\newtheorem{theorem}{Theorem}[section]\crefname{theorem}{Theorem}{Theorems}
\newtheorem{lemma}[theorem]{Lemma}\crefname{lemma}{Lemma}{Lemmas}
\crefname{claim}{Claim}{Claims}
\newtheorem{definition}[theorem]{Definition}\crefname{definition}{Definition}{Definitions}
\crefname{proposition}{Proposition}{Propositions}
\crefname{problem}{Problem}{Problems}
\crefname{remark}{Remark}{Remarks}
\newtheorem{corollary}[theorem]{Corollary}\crefname{corollary}{Corollary}{Corollaries}

\numberwithin{equation}{section}

\DeclareMathOperator*{\argmin}{argmin}

\DeclareMathOperator{\sign}{sign}

\DeclareMathOperator{\poly}{poly}
\DeclareMathOperator{\polylog}{polylog}
\DeclareMathOperator{\diag}{diag}

\renewcommand{\vec}[1]{\bm{#1}}
\newcommand{\mat}[1]{\bm{#1}}
\newcommand{\A}[2]{\mat A(\vec{#1}, \vec{#2})}

\newcommand{\R}{\mathbb R}

\newcommand{\Exp}{\mathbb E}

\newcommand{\grad}{\nabla}
\newcommand{\hess}{\nabla^2}
\newcommand{\eps}{\varepsilon}
\renewcommand{\epsilon}{\varepsilon}
\newcommand{\bigO}[1]{O\left(#1\right)}

\newcommand{\noopsort}[1]{}

\usepackage{xparse}
\NewDocumentCommand{\bigOt}{o m}{%
  \IfNoValueTF{#1}
    {\ensuremath{\widetilde{O}\!\left(#2\right)}}
    {\ensuremath{\widetilde{O}_{#1}\!\left(#2\right)}}%
}

\renewcommand{\vec}[1]{\boldsymbol{\mathbf{#1}}}
\DeclarePairedDelimiter{\abs}{\lvert}{\rvert}

\DeclarePairedDelimiter{\norm}{\lVert}{\rVert}
\DeclarePairedDelimiterX{\ip}[2]{\langle}{\rangle}{#1,#2}
\allowdisplaybreaks[4]

\crefalias{AlgoLine}{line}
\crefname{line}{line}{lines}

\SetAlgoProcName{Subroutine}{Subroutine}

\SetKwInput{Input}{Input}
\SetKwInput{Output}{Output}
\SetKwInput{Guarantee}{Guarantee}
\SetKw{True}{True}
\SetKw{False}{False}

\makeatletter

\newcommand*\wthelper[2]{%
        \hbox{\dimen@\accentfontxheight#1%
                \accentfontxheight#11.3\dimen@
                $\m@th#1\widetilde{#2}$%
                \accentfontxheight#1\dimen@
        }%
}

\newcommand*\accentfontxheight[1]{%
        \fontdimen5\ifx#1\displaystyle
                \textfont
        \else\ifx#1\textstyle
                \textfont
        \else\ifx#1\scriptstyle
                \scriptfont
        \else
                \scriptscriptfont
        \fi\fi\fi3
}
\makeatother

\newcommand{\email}[1]{\href{mailto:#1}{\texttt{#1}}}

\title{Improved quantum lower and upper bounds for matrix scaling}
\author{Sander Gribling\thanks{IRIF, Universit\'e de Paris, CNRS, Paris, France. Partially supported by SIRTEQ-grant QuIPP. \email{gribling@irif.fr}} \and Harold Nieuwboer\thanks{Korteweg--de Vries Institute for Mathematics and QuSoft, University of Amsterdam. Supported by NWO grant OCENW.KLEIN.267. \email{h.a.nieuwboer@uva.nl}}}
\date{}

\begin{document}

\maketitle

\begin{abstract}
Matrix scaling is a simple to state, yet widely applicable linear-algebraic problem: the goal is to scale the rows and columns of a given non-negative matrix such that the rescaled matrix has prescribed row and column sums.
Motivated by recent results on first-order quantum algorithms for matrix scaling, we investigate the possibilities for quantum speedups for classical second-order algorithms, which comprise the state-of-the-art in the classical setting.

We first show that there can be essentially no quantum speedup in terms of the input size in the high-precision regime: any quantum algorithm that solves the matrix scaling problem for $n \times n$ matrices with at most $m$ non-zero entries and with $\ell_2$-error $\varepsilon=\widetilde\Theta(1/m)$ must make $\widetilde\Omega(m)$ queries to the matrix, even when the success probability is exponentially small in $n$.
Additionally, we show that for $\varepsilon\in[1/n,1/2]$, any quantum algorithm capable of producing $\frac{\varepsilon}{100}$-$\ell_1$-approximations of the row-sum vector of a (dense) normalized matrix uses $\Omega(n/\varepsilon)$ queries, and that there exists a constant $\varepsilon_0>0$ for which this problem takes $\Omega(n^{1.5})$ queries.

To complement these results we give improved quantum algorithms in the low-precision regime: with quantum graph sparsification and amplitude estimation, a box-constrained Newton method can be sped up in the large-$\varepsilon$ regime, and outperforms previous quantum algorithms. For entrywise-positive matrices, we find an $\varepsilon$-$\ell_1$-scaling in time $\widetilde O(n^{1.5}/\varepsilon^2)$, whereas the best previously known bounds were $\widetilde O(n^2\mathrm{polylog}(1/\varepsilon))$ (classical) and $\widetilde O(n^{1.5}/\varepsilon^3)$ (quantum).
\end{abstract}

\section{Introduction}

The matrix scaling problem asks to scale each row and column of a given matrix $\mat A \in [0,1]^{n \times n}$ by a positive number in such a way that the resulting matrix has marginals (i.e., row- and column-sums) that are close to some prescribed marginals. For example, one could ask to scale the matrix in such a way that it becomes doubly stochastic. 

Matrix scaling has applications in a wide variety of areas including numerical linear algebra~\cite{lapack}, optimal transport in machine learning~\cite{NIPS2013_af21d0c9}, statistics~\cite{kruithof:telefoon,deming1940least,brown,bishop-fienberg-holland}, and also in more theoretical settings, e.g.~for approximating the permanent~\cite{lsw00}.
For a survey, we refer the reader to~\cite{idel2016review}.
Furthermore, the matrix scaling problem is a special (commutative) instance of a more general (non-commutative) class of problems, which includes operator and tensor scaling; these problems have many more applications and are a topic of much recent interest \cite{garg2019operator,burgisser2019towards}.

Formally, the matrix scaling problem is defined for the $\ell_p$-norm as follows. Given a matrix $\mat A \in [0,1]^{n \times n}$ with at most $m$ non-zero entries, entrywise-positive target marginals $\vec r, \vec c \in \R^{n}$ with $\|\vec r\|_1=1=\|\vec c\|_1$, and a parameter $\eps \geq 0$, find vectors $\vec x, \vec y \in \R^{n}$ such that the (rescaled) matrix $\A x y := (A_{ij} e^{x_i + y_j})_{i,j \in [n]}$ satisfies 
\begin{equation} \label{eq: error def}
\norm{\vec r(\A x y) - \vec r}_p \leq \eps, \qquad \norm{\vec c(\A x y) - \vec c}_p \leq \eps. 
\end{equation}
Here $\vec r(\A x y) = (\sum_{j=1}^n A_{ij} e^{x_i + y_j})_{i \in [n]}$ is the vector of row-marginals of the matrix $\A x y$ and similarly $\vec c(\A x y) = (\sum_{i=1}^n A_{ij} e^{x_i + y_j})_{j \in [n]}$ is the vector of column-marginals.
We refer to $\vec x$ and $\vec y$ as the scaling vectors, whereas $e^{x_i}$ and $e^{y_j}$ are called scaling factors.
A common choice of target marginals is $(\vec r, \vec c) = (\frac{\vec 1}{n}, \frac{\vec 1}{n})$, i.e., every row and column sum target is $1/n$, and we refer to these as the \emph{uniform} target marginals.
As is standard in the matrix scaling literature, we will henceforth assume that $\mat A$ is \emph{asymptotically $(\vec r, \vec c)$-scalable}: for every $\eps > 0$, there exist $\vec x, \vec y$ such that~$\A x y$ satisfies \cref{eq: error def}.
This depends only on the support of~$\mat A$~\cite[Thm.~3]{rothblum&schneider:scaling}, and is the case if and only if $(\vec r, \vec c)$ is in the convex hull of the points $(\vec e_i, \vec e_j) \in \R^{2n}$ such that $A_{ij} > 0$, where the $\vec e_i$ are the standard basis vectors for $\R^n$.
We will also always assume that the smallest non-zero entry of each of $\mat A$, $\vec r$ and $\vec c$ is at least $1/\!\poly(n)$.

Many classical algorithms for the matrix scaling problem can be viewed from the perspective of convex optimization. For example, one can solve the matrix scaling problem by minimizing the convex (potential) function
\begin{equation} \label{def: f}
  f(\vec x, \vec y) = \sum_{i,j=1}^n  A_{ij} e^{x_i + y_j} - \ip{\vec r}{\vec x} - \ip{\vec c}{\vec y},
\end{equation} 
where $\ip{\cdot}{\cdot}$ denotes the standard inner product on $\R^n$.
The popular and practical Sinkhorn algorithm~\cite{sinkhorn:scaling} -- which alternates between rescaling the rows and columns to the desired marginals -- can be viewed as a (block-)coordinate descent algorithm on~$f$, i.e., a first-order method.
Given its simplicity, it is no wonder that it has been rediscovered in many settings, and is known by many names, such as the RAS algorithm, iterative proportional fitting, or raking.

It is known that the iterates in the Sinkhorn algorithm converge to a $(\vec r, \vec c)$-scaled matrix whenever $\mat A$ is asymptotically $(\vec r, \vec c)$-scalable.
The convergence rate of Sinkhorn's algorithm is known in various settings, and we give a brief overview of the (classical) time complexity of finding an $\eps$-$\ell_1$-scaling, noting that a single iteration can be implemented in time $\bigOt{m}$. 
When $\mat A$ is entrywise positive then one can scale in time $\bigOt{n^2 / \eps}$~\cite{qscalingICALP}; in the $\ell_2$-setting for uniform target marginals a similar result can be found in~\cite{KALANTARI1993237,klrs08}.
In the general setting where $\mat A$ has at most $m \leq n^2$ non-zero entries the complexity becomes $\bigOt{m/\eps^2}$ (for arbitrary target marginals $(\vec r,\vec c)$); a proof may be found in \cite{altschuler2017nearlinear} for the entrywise-positive case, \cite{chakrabarty2020better} for exactly scalable matrices (i.e., where the problem can be solved for $\eps=0$) and \cite{qscalingICALP} for asymptotically scalable matrices.

While simple, the Sinkhorn algorithm is by no means the fastest when the parameter~$\eps$ is small.
The classical state-of-the-art algorithms are based on second-order methods such as (traditional) interior point methods or so-called \emph{box-constrained Newton methods}~\cite{cmtv17,azlow17}, the latter of which we describe in more detail below.
We note that these algorithms depend on fast algorithms for graph sparsification and Laplacian system solving, so are rather complicated compared to Sinkhorn's algorithm. 
The box-constrained Newton methods can find $\eps$-$\ell_1$-scaling vectors in time $\bigOt{m R_\infty}$, where the $\widetilde O$ hides polylogarithmic factors in $n$ and $1/\eps$, and $R_\infty$ is a certain diameter bound (made precise later in the introduction).
For entrywise-positive matrices, $R_\infty$ is of size $\bigOt{1}$, and in general it is known to be $\bigOt{n}$ \cite[Lem.~3.3]{azlow17}. Alternatively, the interior-point method of \cite{cmtv17} has a time complexity of $\bigOt{m^{3/2}}$, which is better than the box-constrained Newton method for general inputs, but worse for entrywise-positive matrices.

Recently, a quantum algorithm for matrix scaling was developed based on Sinkhorn's algorithm~\cite{qscalingICALP}, giving $\eps$-$\ell_1$-scaling vectors in time $\bigOt{\sqrt{mn}/\eps^4}$ for general matrices or $\bigOt{n^{1.5} / \eps^3}$ for entrywise-positive matrices. This improves the dependence on $m$ and $n$ at the cost of a higher dependence on $1/\eps$ when compared to the classical Sinkhorn algorithm (which we recall runs in $\bigOt{m/\eps^2}$ or $\bigOt{n^2/\eps}$ for entrywise-positive matrices).
Furthermore, it was shown that this quantum algorithm is optimal for (sufficiently small) constant $\eps$: there exists an $\eps_0 > 0$ (independent of $n$) such that every quantum algorithm that $\eps_0$-$\ell_1$-scales to uniform target marginals with probability at least $2/3$ must make at least $\Omega(\sqrt{mn})$ queries.
It was left as an open problem whether one can also obtain quantum speedups (in terms of $n$ or $m$) using second-order methods. 
In this work we give improved quantum lower and upper bounds on the complexity of matrix scaling.
We first prove a lower bound: we show that every quantum algorithm that solves the matrix scaling problem for small enough $\eps$ must make a number of queries proportional to the number of non-zero entries in the matrix, even when the success probability of the algorithm is only assumed to be exponentially small.
This shows that one cannot hope to get a quantum algorithm for matrix scaling with a polylogarithmic $1/\eps$-dependence and sublinear dependence on $m$.
However, this does not rule out that second-order methods can be useful in the quantum setting.
Indeed, we give a quantum box-constrained Newton method which has a better $1/\eps$-dependence than the previously mentioned quantum Sinkhorn algorithm, and in certain settings is strictly better, such as for entrywise-positive instances.

\subsection{Lower bounds}
As previously mentioned, we show for entrywise-positive instances that a polynomial $1/\eps$-dependence is necessary for a scaling algorithm whose $n$-dependence is $n^{2-\gamma}$ for a constant $\gamma > 0$.
More precisely, we prove the following theorem (which we extend to an $\widetilde\Omega(m)$-lower bound in the general setting of $m\leq n^2$ non-zero entries in \cref{cor: sparse lb}):
\begin{theorem}
\label{thm: informal scaling lb}
There exists a constant $C>0$ such that every matrix scaling algorithm that, with probability $\geq \frac32 \exp(-n/100)$, finds scaling vectors for entrywise-positive $n \times n$-matrices with $\ell_2$-error $C/(n^2 \sqrt{\ln n})$ must make at least $\Omega(n^2)$ queries to the matrix. This even holds for uniform targets and matrices with smallest entry $\Omega(1/n^2)$.
\end{theorem}
The proof of this lower bound is based on a reduction from deciding whether bit strings have Hamming weight $n/2 + 1$ or $n/2 - 1$.
Specifically, given $k$ bit strings $z^1, \ldots, z^k \in \{ \pm 1 \}^n$ for $k = \Theta(n)$, each with Hamming weight $\abs{z^i} = n/2 + a_i$ where $a_i \in \{\pm 1\}$, we show that any matrix scaling algorithm can be used to determine all the $a_i$.
One can show that every quantum algorithm that computes all the $a_i$'s needs to make $\Omega(nk)$ quantum queries to the bit string $z^1, \dotsc, z^k$, even if the algorithm has only exponentially small success probability: to determine a single $a_i$ with success probability at least $2/3$, one needs to make $\Omega(n)$ quantum queries to the bit string $z^i$~\cite{beals2001quantum,nayak1999quantum,ambainis:lowerboundsj}, and one can use the strong direct product theorem of Lee and Roland~\cite{lee&roland:qsdpt} to prove the lower bound for computing all $k$ $a_i$'s simultaneously. 
To convert the problem of computing the $a_i$ to an instance of matrix scaling, one constructs a $2k \times n$ matrix $\mat A$ whose first $k$ rows are (roughly) given by the vectors $1 + z^i / b$ for some $b \geq 2$, and whose last $k$ rows are given by $1 - z^i / b$.
For such an $\mat A$, the column sums are all $2k$, and the row sums are determined by the $a_i$.
If the matrix $\mat A'$ obtained by a single Sinkhorn step from $\mat A$ (i.e., rescaling all the rows) were exactly column scaled, then the \emph{optimal} scaling factors encode the $a_i$.
We show that, if one randomly (independently for each $i$) permutes the $z^i$ beforehand, this is approximately the case: the column sums of this $\mat A'$ will be close to the desired column sums with high probability, and hence the first step of Sinkhorn gives approximately optimal scaling factors (which encode the $a_i$).
Then, we give a lower bound on the strong convexity parameter of the potential $f$, to show that \emph{all} sufficiently precise minimizers of $f$ also encode the~$a_i$.
In other words, from sufficiently precise scaling factors, we can recover the~$a_i$, yielding the reduction to matrix scaling, and consequently a lower bound for the matrix scaling problem.

\bigskip 
We additionally study the problem of computing an $\eps$-$\ell_1$-approximation of the vector of row sums of an $\ell_1$-normalized $n \times n$ matrix $\mat A$.
This is a common subroutine for matrix scaling algorithms; for instance, the gradient of the potential function $f$ from \eqref{def: f} that we optimize for the upper bound can be determined from the row and column sums by subtracting the desired row and column sums, so the complexity of this subroutine directly relates to the complexity of each iteration in our algorithm.
We give the following lower bound for this problem.
\begin{theorem}[Informal] \label{thm: informal grad lb} For $\eps \in [1/n,1/2]$ and an $\ell_1$-normalized matrix $\mat A \in [0,1]^{n \times n}$, computing an $\tfrac{\eps}{100}$-$\ell_1$-approximation of $\vec r(\mat A)$ takes $\Omega(n/\eps)$ queries to $\mat A$. Moreover, there exists a constant $\eps_0>0$ such that computing an $\eps_0$-$\ell_1$-approximation of $\vec r(\mat A)$ takes $\Omega(n^{1.5})$ queries to $\mat A$. 
\end{theorem}
The first lower bound in the theorem is proven in \cref{lb: original}.
Its proof is based on a reduction from $\Theta(n)$ independent instances of the majority problem, as for the lower bound for matrix scaling.
The second lower bound can be derived from the lower bound for matrix scaling given in~\cite{qscalingICALP}: using a constant number of calls to a subroutine that provides constant-precision approximations to the row- and column-sum vectors, one can implement Sinkhorn's algorithm to find a constant-precision $\ell_1$-scaling, which for a small enough constant takes $\Omega(n^{1.5})$ queries. Hence, there exists a constant $\eps_0 > 0$ (independent of $n$) such that computing an $\eps_0$-$\ell_1$-approximation of $\vec r(\mat A)$ takes at least $\Omega(n^{1.5})$ queries to the matrix entries.

\subsection{Upper bounds}
While the first lower bound (\cref{thm: informal scaling lb}) shows that a (quantum) algorithm for matrix scaling cannot have both an $m^{1-\gamma}$-dependence for $\gamma > 0$ and a polylogarithmic $1/\eps$-dependence, one can still hope to obtain a second-order $\bigOt{\sqrt{mn}/\!\poly(\eps)}$-time algorithm with a better $1/\eps$-dependence than the quantum Sinkhorn algorithm of~\cite{qscalingICALP}.
We show that one can build on a box-constrained Newton method~\cite{cmtv17,azlow17} to obtain a quantum algorithm which achieves this, at the cost of depending quadratically on a certain diameter bound $R_\infty$; recall for comparison that the classical box-constrained Newton methods run in time $\bigOt{m R_\infty}$. For general matrices, one has the bound $R_\infty = \bigOt{n}$~\cite[Lem.~3.3]{azlow17}.
The performance of the resulting quantum box-constrained Newton method is summarized in the following theorem:
\begin{theorem}[Informal version of \cref{cor: non-positive,cor: positive}]
  For asymptotically-scalable matrices $\mat A \in \R^{n \times n}_{\geq 0}$ with $m$ non-zero entries and target marginals $(\vec r, \vec c)$, one can find $(\vec x, \vec y)$ such that $\A x y$ is $O(\eps)$-$\ell_1$-scaled to $(\vec r, \vec c)$ in quantum time $\bigOt{R^2_\infty \sqrt{mn}/\eps^2}$ where $R_\infty$ is the $\ell_\infty$-norm of at least one $\eps^2$-minimizer of $f$. When $\mat A$ is entrywise positive we have $R_\infty = \bigOt{1}$, so the algorithm runs in quantum time $\bigOt{n^{1.5}/\eps^2}$.
\end{theorem}
We emphasize that the diameter bound $R_\infty$ does not need to be provided as an input to the algorithm.
Note that for entrywise-positive matrices, the algorithm improves over the quantum Sinkhorn method, which runs in time $\bigOt{n^{1.5}/\eps^3}$.

Let us give a sketch of the box-constrained method that we use, see \cref{sec: minimizing second-order robust} for details. The algorithm aims to minimize the (highly structured) convex potential function $f$ from \cref{def: f}. A natural iterative method for minimizing convex functions $f$ is to minimize in each iteration~$i$ the quadratic Taylor expansion $\frac{1}{2} \vec x^T \nabla^2 f(\vec x^{(i)}) \vec x + \vec x^T \nabla f(\vec x^{(i)}) +  f(\vec x_i)$ of the function at the current iterate. A box-constrained method constrains the minimization of the quadratic Taylor expansion to those $\vec x$ that lie in an $\ell_\infty$-ball of radius $c$ around the current iterate (hence the name):
\[
    \vec x^{(i)} = \argmin_{\norm{\vec x-\vec x^{(i)}}_\infty \leq c} \frac{1}{2} \vec x^T \nabla^2f(\vec x^{(i)}) \vec x + \vec x^T \nabla f(\vec x^{(i)}).
\]
This is guaranteed to decrease a convex function $f$ whenever it is \emph{second-order robust}, i.e., whenever the Hessian of $f$ at a point is a good multiplicative approximation of the Hessian at every other point in a constant-radius $\ell_\infty$-ball. One can show that the steps taken decrease the potential gap by a multiplicative factor which depends on the distance to the minimizer.

One then observes that the function $f$ from \cref{def: f} is second-order robust. Moreover, its Hessian has an exceptionally nice structure: given by
\[
  \nabla^2f(\vec x,\vec y) = \begin{bmatrix} \diag(\vec r(\A x y)) & \A x y \\ {\A x y}^T & \diag(\vec c(\A x y)) \end{bmatrix}\!\!,
\]
it is similar to a \emph{Laplacian} matrix.
This means that the key subroutine in this method (approximately) minimizes quadratic forms $\frac12 \vec z^T \mat H \vec z + \vec z^T \vec b$ over $\ell_\infty$-balls, where $\mat H$ is a Laplacian  matrix; without the $\ell_\infty$-constraint, this amounts to solving the Laplacian system $\mat H \vec z = \vec b$.
Such a subroutine can be implemented for the more general class of symmetric diagonally-dominant matrices (with non-positive off-diagonal entries) on a classical computer in (almost) linear time in the number of non-zero entries of $\mat H$~\cite{cmtv17}. For technical reasons, one has to add a regularization term to $f$, and the regularized potential instead has a symmetric diagonally-dominant Hessian structure.
Given the recent quantum algorithm for graph sparsification and Laplacian system solving of Apers and de Wolf~\cite{Apers2020QLaplacian}, one would therefore hope to obtain a quantum speedup for the box-constrained Newton method.
We show that one can indeed achieve this by first using the quantum algorithm for graph sparsification, and then using the classical method for the minimization procedure. We note, however, that in order to achieve a quantum speedup in terms of $m$ and $n$, we incur a polynomial dependence in the time complexity on the precision with which we can approximate $\mat H$ and $\vec b$ (as opposed to only a \emph{polylogarithmic} dependence classically). Such a speedup with respect to one parameter (dimension) at the cost of a slowdown with respect to another (precision) is more common in recent quantum algorithms for optimization problems and typically requires a more careful analysis of the impact of approximation errors.
Interestingly, for the classical box-constrained Newton method, the minimization subroutine is the bottleneck, whereas in our quantum algorithm, the cost of a single iteration is dominated by the time it takes to approximate the vector $\vec b$. 
Using similar techniques as in~\cite{qscalingICALP}, one can obtain an additive $\delta \cdot \norm{\A x y}_1$-approximation of $\vec b$ in time roughly $\sqrt{mn}/\delta$. To obtain an efficient quantum algorithm we therefore need to control $\norm{\A x y}_1$ throughout the run of the algorithm. We do so efficiently by testing in each iteration whether the $1$-norm of $\A x y$ is  too large, if it is, we divide the matrix by $2$ (by shifting $\vec x$ by an appropriate multiple of the all-ones vector), which reduces the potential. 

\subsection{Open problems}
Our lower bound on matrix scaling shows that it is not possible to provide significant quantum speedups for scaling of entrywise-positive matrices in the high-precision scaling regime.
However, the best classical upper bound for $\eps$-scaling when no assumptions are made on the support of the matrices is $\bigOt{m^{3/2}}$, where $m$ is the number of non-zero entries \cite{cmtv17} (recall that this hides a polylogarithmic dependence on $1/\eps$).
The algorithm that achieves this bound is an interior-point method, rather than a box-constrained Newton method.
It is an interesting open problem whether such an algorithm also admits a quantum speedup in terms of $m$ while retaining a polylogarithmic $1/\eps$-dependence.
Note that while the interior-point method relies on fast Laplacian system solvers, it is not enough to merely replace this by a quantum Laplacian system solver, as the dimension of the linear system in question is $m + n$ rather than $\Theta(n)$.
More generally, the possibility of obtaining quantum advantages in high-precision regimes for optimization problems is still a topic of ongoing investigation.

A second natural question is whether the lower bounds from \cref{thm: informal grad lb} for computing an approximation of the row sums are tight. The best upper bound for the row-sum vector approximation that we are aware of is the one we use in the upper bound for scaling: we can compute an $\eps$-$\ell_1$-approximation of the row- and column sums in time $\bigOt{n^{1.5} / \eps}$.
For constant $\eps_0 \geq \eps > 0$ this matches the lower bound $\Omega(n^{1.5})$ (up to log-factors), but for non-constant $\eps > \frac{1}{100n}$ it remains an interesting open problem to close the gap between $\bigOt{n^{1.5}/\eps}$ and $\Omega(n/\eps)$.

\section{Lower bounds for matrix scaling and marginal approximation}
In this section we prove two lower bounds: an $\widetilde\Omega(m)$-lower bound for $1/\!\poly(n)$-$\ell_2$-scaling $n \times n$ matrices with at most $m$ non-zero entries, and for $\eps \in [1/n, 1/2]$ an $\Omega(n/\eps)$-lower bound for $\eps$-$\ell_1$-approximation of the row-sum vector of a normalized $n \times n$ matrix (with non-negative entries).
The proofs for both lower bounds are based on a reduction from the lower bound given below in \cref{thm:basic lb}.
In \cref{sec: definition} we construct the associated instances for matrix scaling, and in \cref{sec: concentration} we analyze their column marginals after a single iteration of the Sinkhorn algorithm.
Afterwards, in \cref{sec: convexity} we show that these column marginals are close enough to the target marginals for the reduction to matrix scaling to work, and in \cref{sec: concluding the lb} we put the ingredients together, with the main theorem being \cref{thm: scaling lb}.
Finally, in \cref{sec: row marginal lb} we prove the lower bound for computing approximations to the row marginals.

The lower bound we reduce from is the following:
\begin{theorem}
\label{thm:basic lb}
Let $n$ be even, $\tau \in [1/n, 1/2]$ such that $n\tau$ is an integer, and let $k \geq 1$ be an integer. Given $k$ binary strings $z^{1},\ldots,z^{k}\in\{\pm 1\}^n$, where $z^i$ has Hamming weight $n/2+a_i \tau n$ for $a_i\in\{-1,1\}$, computing with probability $\geq \exp(-k/100)$ a string $\tilde{a}\in\{-1,1\}^{k}$ that agrees with $a$ in $\geq 99\%$ of the positions requires $\Omega(k/\tau)$ quantum queries.
\end{theorem}
\begin{proof}
  Let $\mathcal D = \{z \in \{\pm 1\}^n: \abs{z} = n/2 +\tau n \text{ or } \abs{z} = n/2-\tau n\}$ and define the partial Boolean function $f\colon \mathcal D \to \{\pm 1\}$ as 
  \[
  f(z) = \begin{cases} 1 & \text{ if } \abs{z} = n/2+\tau n \\
  -1 & \text{ if } \abs{z} = n/2-\tau n. 
  \end{cases}
  \]
  It is known that computing $f$ with success probability at least $2/3$ takes $\Theta(1/\tau)$ quantum queries to~$z$~\cite[Cor.~1.2]{nayak1999quantum}, i.e., the bounded-error quantum query complexity $Q_{1/3}(f)$ is $\Theta(1/\tau)$.
  
  We now proceed with bounding the query complexity of computing $99\%$ of the entries of $f^{(k)}\colon \mathcal D^k \to \{\pm 1\}^k$ defined by $f^{(k)}(z^1, \dotsc, z^k) = (f(z^1), \dotsc, f(z^k))$.
  We will make use of the general adversary bound $\mathrm{Adv}^{\pm}(f)$~\cite{DBLP:conf/stoc/HoyerLS07} which is known to satisfy $\mathrm{Adv}^{\pm}(f) = \Theta(Q_{1/3}(f))$~\cite[Thm.~1.1]{Lee2011QuantumQC}.
  The strong direct product theorem of Lee and Roland~\cite[Thm.~5.5]{lee&roland:qsdpt} says that for every $0 \leq \delta < 1$, $\mu \in [\frac{1 + \sqrt{\delta}}{2}, 1]$ and integers $k, K$, every quantum algorithm that outputs a bit string~$\tilde a \in \{\pm 1\}^k$, and makes $T$ quantum queries to the bit strings $z^1, \dotsc, z^k$ with \[
    T \leq \frac{k \delta }{K (1 - \delta)} \mathrm{Adv}^{\pm}(f)
  \]
  has the property that $\tilde a$ agrees with $f^{(k)}(z^1, \dotsc, z^k)$ on at least a $\mu$-fraction of the entries 
  with probability at most $\exp(k(\tfrac{1}{K} -  D(\mu \Vert \frac{1 + \sqrt{\delta}}{2})))$.\footnote{In~\cite{lee&roland:qsdpt} the upper bound on $T$ is stated in terms of  $\mathrm{Adv}^*(F)$ where $F = (\delta_{f(x),f(y)})_{x,y \in \mathcal D}$ is the Gram matrix of~$f$. For Boolean functions $f$ one has $\mathrm{Adv}^*(F)=\mathrm{Adv}^{\pm}(f)$~\cite[Thm.~3.4]{Lee2011QuantumQC}.} 
  Here $D(\mu \Vert \frac{1 + \sqrt{\delta}}{2})$ is the Kullback--Leibler divergence between the distributions $(\mu, 1-\mu)$ and $(\frac{1 + \sqrt{\delta}}{2}, \frac{1 - \sqrt{\delta}}{2})$.
  For $\mu = 0.99$, $\delta = 0.1$ and $K = 3$, one has $\frac{1}{K} - D(\mu \Vert \frac{1 + \sqrt{\delta}}{2}) \approx -0.03 \leq -1/100$.
  Therefore, the strong direct product theorem shows that computing $99\%$ of the entries of $f^{(k)}(z^1, \dotsc, z^k) = a$ correctly, with success probability at least $\mathrm{exp}(-k/100)$, takes $\Omega(k\, \mathrm{Adv}^{\pm}(f)) = \Omega(k\, Q_{1/3}(f)) = \Omega(k/\tau)$ quantum queries.
\end{proof}
We will use this lower bound with $k=n/2$ and $\tau=1/n$. The following intuition is useful to keep in mind. For a fixed $b \geq 2$, define the $2k \times n$ matrix $\mat A$ whose $(2i-1)$-th row equals $1+z^i/b$ and whose $(2i)$-th row equals $1-z^i/b$. Then $\mat A$ has the property that the row-marginals encode the Hamming weights of the $z^i$, and are all very close to $n$. (This implies that the first row-rescaling step of Sinkhorn's algorithm encodes the $a_i$.) Moreover, the column-marginals are exactly uniform. Hence, one may hope that all sufficiently precise scalings of $\mat A$ to uniform targets have scaling factors that are close to those given by the first row-rescaling step of Sinkhorn's algorithm (and hence learn most of the $a_i$). 

Below we formalize this approach. We show that if one randomly permutes the coordinates of each $z^i$ (independently over $i$), then with high probability, all $\eps$-scalings of the resulting matrix $\mat A^\sigma$ are close to the first step of Sinkhorn's algorithm; here we need to choose $b$ sufficiently large (${\sim}\,\sqrt{\ln(n)}$) and $\eps$ sufficiently small (${\sim}\,\frac{1}{n^2b}$). The section is organized as follows. In \cref{sec: definition} we formally define our matrix scaling instances and we analyse the first row-rescaling step of Sinkhorn's algorithm. In \cref{sec: concentration} we show that after the row-rescaling step, with high probability (over the choice of permutations), the column-marginals are close to uniform. In \cref{sec: convexity,sec: concluding the lb} we use the strong convexity of the potential $f$ from \cref{def: f} to show that if the above event holds, then all approximate minimizers of $f$ can be used to solve the counting problem.

\subsection{Definition of the scaling instances and analysis of row marginals} \label{sec: definition}
Let $n \geq 4$ be even. Let $k =n/2$ and let $z^1, \dotsc, z^k \in \{\pm 1 \}^n$ have Hamming weight $\abs{z^i} = \abs{\{ j : z_j^i = 1 \}} = n/2 + a_i$ for $a_i \in \{\pm 1\}$.
Sample uniformly random permutations $\sigma^1, \dotsc, \sigma^k \in S_n$ and define $w^i$ by $w^i_j = \smash{z^i_{(\sigma^i)^{-1}(j)}}$.
Let $b \geq 2$ be some number depending on $n$, and consider the $2k \times n$ matrix $\mat A^\sigma$ whose entries are $\mat A^\sigma_{2i-1,j} = 1 + \frac{w^i_j}{b}$ and $\mat A^\sigma_{2i,j} = 1 - \frac{w^i_j}{b}$.
Then each column sum $c_j(\mat A^\sigma)$ is $2k$, and the row sums of $\mat A^\sigma$ are given by
\[
    r_{2i-1}(\mat A^\sigma) = n + \frac1b \sum_{j=1}^n w^i_j = n + \frac2b a_i, \quad r_{2i}(A^\sigma) = n - \frac2b a_i.
\]
Let 
\begin{equation} \label{eq: first step}
    X_{2i-1} = \frac{1}{2k} \cdot \frac{1}{n + \frac2b a_i} \text{ and } X_{2i} = \frac{1}{2k} \cdot \frac{1}{n - \frac2b a_i} \qquad \text{ for all } i \in [k]
\end{equation}
be the row scaling factors obtained from a single Sinkhorn step. We first observe that the difference between $x_{2i-1} := \ln(X_{2i-1})$ and $x_{2i} := \ln(X_{2i})$ permits to recover $a_i$.
\begin{lemma} \label{lem: row marginals step 1}
    For the specific row-scaling factors $\mat X$ for $\mat A^\sigma$ given in \eqref{eq: first step}, for every $i \in [k]$ it holds that
    \[
      \abs*{\ln(X_{2i-1} / X_{2i})} \geq \frac{4}{nb},
    \]
     and $\sign(\ln(X_{2i}/X_{2i-1})) = a_i$.
\end{lemma}
\begin{proof}
    Observe that ($nb>2$ and therefore)
    \begin{equation*}
        \abs*{\ln(X_{2i-1}/X_{2i})} = \abs*{\ln\left(\frac{n + \frac{2}{b}}{n - \frac{2}{b}}\right)} = \ln\left(\frac{nb + 2}{nb - 2}\right) \geq \frac{4}{nb}. \qedhere
    \end{equation*}
\end{proof}

\subsection{Concentration of column marginals} \label{sec: concentration}
We first give an explicit expression for the $j$th column marginal of $\mat X\mat A^\sigma$ where $\mat X$ is given in \eqref{eq: first step}. 
\begin{lemma}
    The matrix $\mat X \mat A^\sigma$ has column sums
    \[
        c_j(\mat X \mat A^\sigma) = \frac{1}{2k (n^2 - 4/b^2)} \left(2kn - \frac{4}{b^2}\sum_{i=1}^{k} w_j^i a_i\right) \quad \text{ for } j \in [n].
    \]
\end{lemma}
\begin{proof}
    We have
    \begin{align*}
        c_j(\mat X \mat A^\sigma) & = \sum_{i=1}^{k} \left(\frac{1+w^i_j/b}{2k(n+2a_i/b)}+ \frac{1-w^i_j/b}{2k(n-2a_i/b)}\right) \\
        & = \frac{1}{2k(n^2 - 4/b^2)} \sum_{i=1}^{k} \left((1+w^i_j/b)(n-2a_i/b) + (1-w^i_j/b)(n+2a_i/b)\right) \\
        & = \frac{1}{2k(n^2 - 4/b^2)} \sum_{i=1}^{k} \left(2n - \frac{4w_j^i a_i}{b^2}\right) \\
        & = \frac{1}{2k(n^2 - 4/b^2)} \left(2kn - \frac{4}{b^2}\sum_{i=1}^{k} w^i_j a_i\right). \qedhere
    \end{align*}
\end{proof}

We now show that with high probability (over the choice of permutations) the column marginals are close to uniform. To do so, we first compute the expectation of $\sum_{i=1}^{k} w_j^i a_i$ (\cref{cor: expected col marginal}). This quantity allows us to obtain the desired concentration of the column marginals via Hoeffding's inequality (\cref{lem: concentration of col marginals}). 
\begin{lemma}
    Let $I = \{ i \in [k] : a_i = 1 \}$ and $I^c = [k] \setminus I$. Define random variables $W_j$, $W_j^c$ by
    \[
        W_j = \sum_{i \in I} w^i_j, \quad W_j^c = \sum_{i \in I^c} w^i_j.
    \]
    Then $\Exp[W_j] = \frac{2 \abs{I}}{n}$ and $\Exp[W_j^c] = -\frac{2 \abs{I^c}}{n}$.
\end{lemma}
\begin{proof}
    Observe that each $w^i_j$ is $1$ with probability $\frac{1}{2} + \frac{a_i}{n}$ because $\sigma^i$ is chosen uniformly randomly from $S_n$, and is $-1$ with probability $\frac{1}{2} - \frac{a_i}{n}$. Therefore $\Exp[w^i_j] = \frac{2a_i}{n}$. By linearity of expectation, the result follows.
\end{proof}
\begin{corollary} \label{cor: expected col marginal}
    We have
    \[
        \Exp\left[\sum_{i=1}^k w^i_j a_i \right] = \Exp[W_j] - \Exp[W_j^c] = \frac{2 (\abs{I} + \abs{I^c})}{n} = \frac{2k}{n}. 
    \]
\end{corollary}
\begin{lemma}
    \label{lem: concentration of col marginals}
    For $t \geq 0$ and $j \in [n]$, with probability at least $1 - 2e^{-t^2 / 2}$, we have
    \[
         \abs*{c_j(\mat X \mat A^\sigma) - \frac1n} = \bigO{\frac{t}{b^2 n^2 \sqrt{k}}}.
    \]
\end{lemma}
\begin{proof}
    Observe first that
    \begin{align*}
      \abs*{c_j(\mat X \mat A^\sigma) - \frac1n} & = \abs*{ \frac{1}{2k(n^2 - 4/b^2)} \left(2kn - \frac{4}{b^2}\sum_{i=1}^{k} w^i_j a_i\right) -\frac1n} \\
      & = \frac{1}{2k n (n^2 - 4/b^2)} \abs*{ n \left(2kn - \frac{4}{b^2}\sum_{i=1}^{k} w^i_j a_i\right) - 2k(n^2 - \frac{4}{b^2})} \\
      & = \frac{1}{2k n (n^2 - 4/b^2)} \abs*{ \frac{8k}{b^2} - \frac{4n}{b^2}\sum_{i=1}^{k} w^i_j a_i} \\
      & = \frac{4}{2k n(n^2-4/b^2)b^2} \abs*{ 2k - n \sum_{i=1}^{k} w^i_j a_i}.
    \end{align*}
    For fixed $j$ and distinct $i,i' \in [k]$, $w^i_j$ and $w^{i'}_j$ are independently distributed random variables because $\sigma^i$ and $\sigma^{i'}$ are independent. Therefore, $V_j := W_j - W_j^c = \sum_{i=1}^k w^i_ja_i$ is a sum of $k$ independent random variables, with each $a_i w^i_j \in [-1, 1]$, and Hoeffding's inequality yields for any $t \geq 0$ that
    \[
        \Pr[\abs{V_j - \Exp[V_j]} \geq t \cdot \sqrt{k}] \leq 2 \exp(-t^2 / 2).
    \]
    Assuming that $\abs{V_j - \Exp[V_j]} \leq t \sqrt{k}$, we have
    \begin{align*}
        \abs*{2k - n \sum_{i=1}^k a_i w^i_j} & = n \abs*{\Exp[V_j] - V_j} \leq n t \sqrt{k}.
    \end{align*}
    With this estimate, we see that
    \begin{equation*}
        \abs*{c_j(\mat X \mat A^\sigma) - \frac1n} \leq \frac{4}{2k n(n^2-4/b^2)b^2} \cdot n t \sqrt{k} = \frac{2 t}{b^2 (n^2 - 4/b^2) \sqrt{k}}. \qedhere
    \end{equation*}
\end{proof}
\begin{corollary} \label{cor: grad bound}
    For any $t \geq 0$, with probability $\geq 1 - 2n e^{-t^2/2}$, we have
    \begin{align*}
        \norm*{c(\mat X \mat A^\sigma) - \frac{\vec 1}{n}}_2 \leq \frac{2 \sqrt{n} t}{b^2 (n^2 - 4/b^2) \sqrt{k}} = \bigO{\frac{t}{b^2 n^2}}.
    \end{align*}
\end{corollary}

\subsection{Strong convexity properties of the potential} \label{sec: convexity}

For a $\lambda$-strongly convex function $f$, the set $\{\vec z : \|\nabla f(\vec z)\|_2 \leq \eps\}$ has a diameter that is bounded by a function of $\lambda$ (we make this well-known fact precise in \cref{lem:strong convexity dist to minimizer}). We show that our potential is strongly convex when viewed as a function from (a suitable subset of) the linear subspace $V = \{(\vec x, \vec y) \in \R^n \times \R^n : \ip{(\vec x, \vec y)}{(\vec 1_n, - \vec 1_n)} = 0\}$ to $\R$ (note that $f$ is invariant under translation by multiples of $(\vec 1_n, -\vec 1_n)$).
We use this to show that whenever $\norm{\grad f(\vec x, \vec y)}_2$ is small, $(\vec x, \vec y)$ is close to the minimizer of $f$ on $V$.
It is easy to verify that \cref{cor: grad bound} in fact gives an upper bound on the norm of the gradient at $(\ln(\mat X),\vec 0)$ (with $\mat X$ as in \eqref{eq: first step}). This implies that $(\ln(\mat X),\vec 0)$ is close to the minimizer of $f$ on $V$, and by the triangle inequality, is also close to any other $(\vec x, \vec y)$ for which $\norm{\grad f(\vec x, \vec y)}_2$ is small.
In the rest of this section we make the above precise. 

In \cref{lem: Hessian lb} we show that the Hessian of $f$ restricted to $V$ has smallest eigenvalue at least~$n\cdot\mu(\vec x, \vec y)$ where $\mu(\vec x, \vec y)$ is the smallest entry appearing in $(A_{ij} e^{x_i+y_j})_{i,j}$.
In \cref{lem: smallest entry lb} we show that $\mu(\vec x^*, \vec y^*) = \Theta(1/n^2)$. This implies that $\mu(\vec x, \vec y) = \Theta(1/n^2)$ for all $(\vec x, \vec y)$ that are a constant distance away from $(\vec x^*, \vec y^*)$ in the $\ell_\infty$-norm, in other words, $f$ is $\Theta(1/n)$-strongly convex around its minimizer. \cref{lem: grad to inf norm bound} summarizes these lemmas: it gives a quantitative bound on the distance to a minimizer, in terms of the gradient. 

\begin{lemma} \label{lem: Hessian lb}
    Let $\mat A$ be an entrywise non-negative $n \times n$ matrix with $\|\mat A\|_1 = 1$ and let $f\colon V \subset \R^n \times \R^n \to \R$ be the potential for this matrix as given in \eqref{def: f}, where $V$ is the orthogonal complement of $(\vec 1_{n}, -\vec 1_n)$.
    Then $\hess f(\vec x, \vec y) \succeq \mu(\vec x, \vec y) \cdot n \cdot \mat P_V$ where $\mat P_V$ is the projection onto $V$ and $\mu(\vec x, \vec y)$ is the smallest entry appearing in $\A x y $. In particular, $f$ is strictly convex on $V$.
\end{lemma}
\begin{proof}
The Hessian of the potential $f(\vec x, \vec y) = \sum_{i,j=1}^n A_{ij} e^{x_i + y_j} - \ip{\vec r}{\vec x} - \ip{\vec c}{\vec y}$ is given by
    \begin{equation*}
        \hess f(\vec x, \vec y) = \begin{bmatrix}
        \diag(\vec r(\A x y) & \A x y \\
        \A x y^T & \diag(\vec c(\A x y))
        \end{bmatrix}.
    \end{equation*}
    We give a lower bound on the non-zero eigenvalues of the Hessian as follows. 
    Conjugating the Hessian with the $2n \times 2n$ matrix $\diag(\mat I, -\mat I)$ preserves the spectrum and yields the matrix 
    \begin{equation*}
        \begin{bmatrix}
            \diag(\vec r(\A  x y) & -\A x y \\
            -\A x y^T & \diag(\vec c(\A  x y))
        \end{bmatrix}.
    \end{equation*}
    One can recognize this as a weighted Laplacian of a complete bipartite graph.
    We denote by $\mu(\vec x, \vec y)$ the smallest entry of $\A x y$ and we use $\mat J$ for the $n \times n$ all-ones matrix. Then
    \begin{equation*}
        \begin{bmatrix}
            \diag(\vec r(\A  x y) & -\A x y \\
            -\A x y^T & \diag(\vec c(\A  x y))
        \end{bmatrix}
        \succeq
        \begin{bmatrix}
            n \mu(\vec x, \vec y) \mat I & - \mu(\vec x, \vec y) \mat J \\
            - \mu(\vec x, \vec y) \mat J & n \mu(\vec x, \vec y) \mat I
        \end{bmatrix}
        = \mu(\vec x, \vec y) \begin{bmatrix}
            n \mat I & -\mat J \\
            -\mat J & n \mat I
        \end{bmatrix},
    \end{equation*}
    where the PSD inequality follows because the difference of the terms is the weighted Laplacian of the bipartite graph with weighted bipartite adjacency matrix $\A x y - \mu(\vec x, \vec y) \mat J$, which has non-negative entries.
    Now observe that the last term $\begin{bmatrix}
            n \mat I & -\mat J \\
            -\mat J & n \mat I
        \end{bmatrix}$ is the (unweighted) Laplacian of the complete bipartite graph $K_{n,n}$, whose spectrum is $2n$, $n$, $0$ with multiplicities $1$, $2n-2$ and $1$ respectively.
    The zero eigenvalue corresponds to the all-ones vector of length $2n$ and it is easy to see that indeed $(\vec 1, -\vec 1)$ also lies in the kernel of $\nabla^2 f(\vec x, \vec y)$.
    This shows that the non-zero eigenvalues of $\hess f(\vec x, \vec y)$ are at least $n \cdot \mu(\vec x, \vec y)$, and that it has a one-dimensional eigenspace corresponding to $0$, spanned by the vector $(\vec 1, -\vec 1)$. Hence, $\hess f(\vec x, \vec y) \succeq \mu(\vec x, \vec y) \cdot n \cdot \mat P_V$. 
\end{proof}

We now bound the smallest entry of the rescaled matrix. For this we use the following lemma (cf.~\cite[Lem.~6.2]{klrs08},~\cite[Cor.~C.3 (arXiv)]{qscalingICALP}) which bounds the variation norm of the scaling vectors $(\vec x^*, \vec y^*)$ of an exact scaling. 
\begin{lemma}\label{cor: variation norm bound}
  Let $\mat A \in [\mu, \nu]^{n \times n}$ and let $(\vec x^*, \vec y^*) \in \R^n \times \R^n$ be such that $\mat A(\vec x^*, \vec y^*)$ is exactly $(\vec r, \vec c)$-scaled.
  Then
  \[
    x_{\max}^* - x_{\min}^* \leq \ln \frac{\nu}{\mu} + \ln \frac{r_{\max}}{r_{\min}},
  \]
  and
  \[
    y_{\max}^* - y_{\min}^* \leq \ln \frac{\nu}{\mu} + \ln \frac{c_{\max}}{c_{\min}}.
  \]
\end{lemma}
\begin{lemma} \label{lem: smallest entry lb}
Let $\mat A \in [\mu,\nu]^{n \times n}$ be an entrywise-positive matrix with $\|\mat A\|_1 = 1$ and let $f\colon V \subset \R^{n} \times \R^n \to \R$ be the potential for this matrix as given in \eqref{def: f}, where $V$ is the orthogonal complement of $(\vec 1_{n}, -\vec 1_n)$. Let $(\vec x^*, \vec y^*) \in V$ be the unique minimizer of $f$ in $V$. Then $\mu(\vec x^*, \vec y^*) \geq \frac{1}{n^2 } \left(\frac{\mu}{\nu}\right)^3$. Moreover, for any $(\vec x, \vec y) \in V$ we have $\mu(\vec x, \vec y) \geq \mu(\vec x^*, \vec y^*)  e^{-2\norm{(\vec x, \vec y) - (\vec x^*, \vec y^*)}_\infty}$.
\end{lemma}
\begin{proof}
By \cref{lem: Hessian lb} $f$ is strictly convex on $V$. We also know that $\mat A$ is exactly scalable. Hence $f$ has a unique minimizer $(\vec x^*, \vec y^*)$. By \cref{cor: variation norm bound} we know that the variation norm of $x^*$ and $y^*$ are bounded by $\ln(\nu/\mu)$. Hence, for every $i,i',j,j' \in [n]$ we have 
\[
\abs*{\ln\left(\frac{e^{x_i^* + y_j^*}}{e^{x_{i'}^* + y_{j'}^*}}\right)} \leq \abs{x_i^* - x_{i'}^*}  + \abs{y_j^* - y_{j'}^*} = 2\ln(\nu/\mu).
\]
Therefore, the ratio between entries of $\mat A(\vec x^*, \vec y^*)$ is bounded:
\[
\left|\frac{\mat A(\vec x^*, \vec y^*)_{ij}}{\mat A(\vec x^*, \vec y^*)_{i'j'}}\right| \leq \left|\frac{A_{ij}}{A_{i'j'}}\right| \left|\left(\frac{e^{x_i^* + y_j^*}}{e^{x_{i'}^* + y_{j'}^*}}\right)\right|
\leq \frac{\nu}{\mu} e^{2\ln(\nu/\mu)} = \left(\frac{\nu}{\mu}\right)^3.
\]
Since the sum of the entries of $\mat A(\vec x^*, \vec y^*)$ equals $1$, this implies that the smallest entry of $\mat A(\vec x^*, \vec y^*)$ is at least $\mu(\vec x^*, \vec y^*) \geq \frac{1}{n^2} \left(\frac{\mu}{\nu}\right)^3$.
Finally, for any $(\vec x, \vec y) \in V$ and any $i,j \in [n]$ we have 
\[
A_{ij}e^{x_i+y_j}\geq A_{ij}e^{x_i^*+y_j^* -2\norm{(\vec x, \vec y) - (\vec x^*, \vec y^*)}_\infty}
\]
which shows $\mu(\vec x, \vec y) \geq \mu(\vec x^*, \vec y^*)  e^{-2\norm{(\vec x, \vec y) - (\vec x^*, \vec y^*)}_\infty}$.
\end{proof}
Finally, to obtain a diameter bound for the set of points with a small gradient we will use the following (well-known) lemma. 
\begin{lemma}\label{lem:strong convexity dist to minimizer}
Assume $g\colon\R^d\to\R$ is a $C^2$ convex function such that $\grad g(\vec 0) = 0$, and assume that for all $\vec x \in \R^d$ with $\norm{\vec x}_\infty \leq r$, we have $\hess g(\vec x) \succeq \lambda I$.
Then 
\begin{equation*}
    \norm{\grad g(\vec x)}_2 \geq \lambda \norm{\vec x}_2 \min(1, r / \norm{\vec x}_\infty) \geq \lambda \min(\norm{\vec x}_\infty, r).
\end{equation*}
In particular, to guarantee that $\norm{\vec x}_\infty \leq C$ for $C \geq 0$, it suffices to show that $\norm{\grad g(\vec x)}_2 < \lambda \min(C, r)$ (note that the strict inequality is necessary here because it forces $\min(\norm{\vec x}_\infty, r) = \norm{\vec x}_\infty$).
\end{lemma}
\begin{proof}
  Fix $\vec x \in \R^n$ and consider $h\colon\R\to\R$ defined by $h(t) = g(t \vec x)$.
  Then $h$ is convex, $\partial_{t=0} h(t) = 0$ and $\partial^2_{t=s} h(t) \geq 0$ for all $s \in \R$. 
  Now assume for $s \in \R$ that $\abs{s} \norm{\vec x}_\infty \leq r$. Then
  \begin{align*}
      \partial_{t=s}^2 h(t) & = \partial_{t=s} (Dg(t\vec x)[\vec x]) \\
      & = D^2g(s \vec x)[\vec x, \vec x] =  \vec x^T \hess g(s\vec x) \vec x \geq \lambda \norm{\vec x}_2^2.
  \end{align*}
  We use this to further estimate, for $s \geq 0$, that
  \begin{align*}
      \ip{\grad g(s\vec x)}{\vec x} = \partial_{t=s} h(t) & = \int_0^s \partial^2_{t=\tau} h(t) \, d \tau \\
      & \geq \int_0^{\min(s, r / \norm{\vec x}_\infty)} \partial^2_{t=\tau} h(t) \, d \tau \\
      & \geq \int_0^{\min(s, r / \norm{\vec x}_\infty)} \lambda \norm{\vec x}_2^2 \, d \tau \\
      & = \lambda \norm{\vec x}_2^2 \min(s, r / \norm{\vec x}_\infty),
  \end{align*}
  where the first inequality follows from the convexity of $h$.
  Setting $s = 1$ and using the Cauchy--Schwarz inequality gives
  \begin{align*}
      \norm{\grad g(\vec x)}_2 \norm{\vec x}_2 \geq \lambda \norm{\vec x}_2^2 \min(1, r / \norm{\vec x}_\infty)
  \end{align*}
  so
  \begin{align*}
    \norm{\grad g(\vec x)}_2 & \geq \lambda \norm{\vec x}_2 \min(1, r / \norm{\vec x}_\infty) \\
    & \geq \lambda \norm{\vec x}_\infty \min(1, r / \norm{\vec x}_\infty) \\
    & = \lambda \min(\norm{\vec x}_\infty, r). \qedhere
  \end{align*}
\end{proof}

\begin{lemma} \label{lem: grad to inf norm bound}
Let $\mat A \in [\mu,\nu]^{n \times n}$ be an entrywise non-negative matrix with $\|\mat A\|_1 = 1$ and let $f\colon V \subset \R^{n} \times \R^n \to \R$ be the potential for this matrix as given in \eqref{def: f}, where $V$ is the orthogonal complement of $(\vec 1_{n}, -\vec 1_n)$. Let $(\vec x^*, \vec y^*)$ be the unique minimizer of $f$ in $V$ and let $0<\delta<1$. Let $(\vec x, \vec y) \in V$ be such that $\|\nabla f(\vec x, \vec y)\|_2 < \delta \cdot \frac{1}{n} \left(\frac{\mu}{\nu}\right)^3 e^{-2}$. Then $\|(\vec x, \vec y) - (\vec x^*, \vec y^*)\|_\infty \leq \delta$.
\end{lemma}
\begin{proof}
\cref{lem: Hessian lb} shows that $\nabla^2 f(\vec x, \vec y) \succeq n \cdot \mu(\vec x, \vec y) \cdot \mat P_V$, where $\mat P_V$ is the orthogonal projector on $V$. \cref{lem: smallest entry lb} shows that $\mu(\vec x, \vec y) \geq \mu(\vec x^*, \vec y^*)  e^{-2\norm{(\vec x, \vec y) - (\vec x^*, \vec y^*)}_\infty} \geq  \frac{1}{n^2 } \left(\frac{\mu}{\nu}\right)^3 e^{-2\norm{(\vec x, \vec y) - (\vec x^*, \vec y^*)}_\infty}$. Hence, for $(\vec x, \vec y)$ with $\|(\vec x, \vec y)- (\vec x^*, \vec y^*)\|_\infty \leq 1$, we have $\nabla^2 f(\vec x, \vec y) \succeq \frac{1}{n}\left(\frac{\mu}{\nu}\right)^3 e^{-2} \cdot \mat P_V$. It then follows from \cref{lem:strong convexity dist to minimizer} that if $\|\nabla f(\vec x, \vec y)\|_2 < \delta \cdot \frac{1}{n}\left(\frac{\mu}{\nu}\right)^3 e^{-2}$, then $\|(\vec x, \vec y)-(\vec x^*, \vec y^*)\|_\infty \leq \delta$.
\end{proof}
Observe that for $\mat A^\sigma$ the ratio between its largest and smallest entry is $\frac{b+1}{b-1} \leq 3$. This gives the following corollary. 
\begin{corollary} \label{cor: grad bound A sigma}
Let $\mat A^\sigma$ be as in \cref{sec: definition} and let $f$ be the associated potential. Let $(\vec x^*, \vec y^*)$ be the unique exact scaling of $\mat A^\sigma$ in $V$. If $(\vec x, \vec y) \in V$ is such that $\|\nabla f(\vec x, \vec y)\|_2 < \frac{\delta}{27n e^2}$, then $\norm{(\vec x, \vec y)-(\vec x^*, \vec y^*)}_\infty \leq \delta$.
\end{corollary}

\subsection{Concluding the lower bound for matrix scaling} \label{sec: concluding the lb}
Let $(\vec{\bar x}, \vec{\bar y}) \in V$ be the unique vector such that $(\vec{\bar x}, \vec{\bar y}) - (\vec x, \vec y)$ is a multiple of $(\vec 1_n, - \vec 1_n)$, where $(\vec x, \vec y)$ are the scaling vectors of the first step of Sinkhorn.
By choosing $t$ and $b$ appropriately we obtain, with high probability over the choice of permutations, a bound on the distance between $(\vec{\bar x},\vec{\bar y})$ and the unique scaling vectors $(\vec x^*, \vec y^*) \in V$ of an exact scaling of $\mat A^\sigma$. This allows us to conclude that, with high probability, all sufficiently precise scalings of $\mat A^\sigma$ encode the Hamming weights $a_i$.  
\begin{corollary}
\label{cor:scaling encodes hamming weights}
    There exists a constant $C>0$ such that for $b = C \sqrt{\ln n}$  the following holds. 
    With probability $\geq 2/3$ (over the choice of $\sigma$) we have for the exact scaling vectors $(\vec x^*, \vec y^*) \in V$ of $\mat A^\sigma$ that 
    \[
    a_i = \sign(x^*_{2i} - x^*_{2i-1}) \quad \text{ for all } i.
    \]
      Furthermore, there exists a constant $C'>0$ such that  for any $(x', y')$ that yield a $(C'/n^2 b)$-$\ell_2$-scaling of $\mat A^\sigma$, $a_i$ can be recovered from $x'$ as $a_i = \sign(x_{2i}-x_{2i-1}) = \sign (x_{2i}' - x_{2i-1}')$ .
\end{corollary}
\begin{proof}
Applying \cref{cor: grad bound} with $t = 10 \sqrt{\ln n}$ shows that with probability at least $2/3$ we have $\norm{\grad f(\vec{\bar x},\vec{\bar y})}_2 = \norm{\grad f(\vec x, \vec y)}_2 = \frac{t}{b}\frac{2 \sqrt{n}}{b (n^2 - 4/b^2) \sqrt{k}}$. Hence, there exists a constant $C>0$ such that for $b=C t$ we have 
\[
\norm{\grad f(\vec{\bar x},\vec{\bar y})}_2 \leq \frac{1}{nb} \frac{1}{27ne^2}.
\]
\cref{cor: grad bound A sigma} then implies that $\norm{(\vec{\bar x},\vec{\bar y})-(\vec x^*, \vec y^*)}_\infty \leq \frac{1}{nb}$ and hence $\abs{(x^*_{2i-1} - x^*_{2i}) - (x_{2i-1} - x_{2i})} \leq \frac{2}{nb}$. Together with \cref{lem: row marginals step 1} (which shows that $\abs{x_{2i-1} - x_{2i}} \geq \frac{4}{nb}$) this means that $a_i = \sign(x^*_{2i} - x^*_{2i-1})$. Moreover, $\abs{x^*_{2i-1} - x^*_{2i}} \geq \frac{2}{nb}$. 

Now consider approximate scalings of $\mat A^\sigma$. Without loss of generality we may assume that the $(x',y')$ that yield a $(\frac{1}{2nb} \frac{1}{27ne^2})$-$\ell_2$-scaling of $\mat A^\sigma$ belong to $V$ (otherwise we shift it by an appropriate multiple of $(\vec 1_n, -\vec 1_n)$). Then, again due to \cref{cor: grad bound A sigma}, we obtain that  $\norm{(\vec x', \vec y')-(\vec x^*, \vec y^*)}_\infty \leq \frac{1}{2nb} \leq \frac{1}{4}\abs{x^*_{2i-1} - x^*_{2i}}$ and hence $\abs{(x'_{2i-1} - x'_{2i}) - (x^*_{2i-1} - x^*_{2i})} \leq \frac{1}{2}\abs{x^*_{2i-1} - x^*_{2i}}$ which means that $\sign(x'_{2i}-x'_{2i-1}) = \sign(x^*_{2i-1} - x^*_{2i}) = a_i$.
\end{proof}

\begin{theorem} \label{thm: scaling lb}
There exists a constant $C>0$ such that any matrix scaling algorithm that, with probability $\geq \frac32 \exp(-n/100)$, finds scalings for $n \times n$-matrices with $\ell_2$-error $C/(n^2 \sqrt{\ln n})$ must make at least $\Omega(n^2)$ queries to the matrix. This even holds for uniform targets and entrywise-positive matrices with smallest entry $\Omega(1/n^2)$.
\end{theorem}
\begin{proof}
We construct a set of hard instances as in \cref{sec: definition}. Let $n \geq 4$ be even. Let $k =n/2$ and let $z^1, \dotsc, z^k \in \{\pm 1 \}^n$ have Hamming weight $\abs{z^i} = \abs{\{ j : z_j^i = 1 \}} = n/2 + a_i$ for $a_i \in \{\pm 1\}$. By \cref{thm:basic lb}, finding at least $99\%$ of the $a_i$'s with probability $\geq \exp(-n/100)$ takes $\Omega(n^2)$-queries to the $z^i_j$. 
One can recover the $a_i$'s with probability $\geq 2/3$ as follows. First, sample the $\sigma^1, \dotsc, \sigma^{n/2}$ uniformly from $S_n$.
A single query to $\mat A^\sigma$ takes a single query to some $w^i$, which takes a single query to $z^i$.
Using \cref{cor:scaling encodes hamming weights}, there exists a constant $C>0$ such that, with probability $\geq 2/3$, any scaling of $\mat A^\sigma$ with $\ell_2$-error $C/(n^2 \sqrt{\ln n})$ recovers \emph{all} $a_i$'s.
Therefore any matrix scaling algorithm finding such a scaling with probability $\geq \exp(-n/100)$ allows us to find all $a_i$'s with probability $\geq \exp(-n/100)$.
\end{proof}

\begin{corollary}
\label{cor: sparse lb}
There exist constants $C_0,C_1>0$ such that any matrix scaling algorithm that, with probability $\geq \exp(-C_0 n/\ln(n))$, finds scalings for $n \times n$-matrices with at most $m$ non-zero entries and $\ell_2$-error $C_1/(m \sqrt{\ln(m/n)})$ must make at least $\widetilde \Omega(m)$ queries to the matrix. This even holds for uniform targets and matrices with smallest non-zero entry $\Omega(1/m)$.
\end{corollary}
\begin{proof}
We construct a set of sparse hard instances by taking direct sums of the hard instances used in the proof of \cref{thm: scaling lb}. Concretely, let $s \geq 4$ be even and assume that $n$ is a multiple of $s$. Let $\mat A_1^{\sigma_1},\ldots, \mat A_{n/s}^{\sigma_{n/s}} \in [0,1]^{s \times s}$ be $n/s$ independently drawn hard instances from the set constructed in the proof of \cref{thm: scaling lb}. We use this to create a sparse instance $\mat A = \frac{s}{n}\oplus_{i=1}^{n/s} \mat A_i^{\sigma_i}$. Note that $\|\mat A\|_1=1$ and each row of $\mat A$ has exactly $s$ non-zero entries (which means $m=ns$). Let $(\vec x, \vec y)$ be an $\eps$-$\ell_2$-scaling of $\mat A$ to uniform marginals. Then we have 
\[
\eps^2 \geq \|\vec 1/n - r(\A x y)\|_2^2  = \sum_{i=1}^{n/s} \|\vec 1/n - \frac{s}{n} \mat A_i^{\sigma_i}(\vec x|_{i},\vec y|_{i})\|_2^2 =  \sum_{i=1}^{n/s} (\frac{s}{n})^2 \|\vec 1/s - \mat A_i^{\sigma_i}(\vec x|_{i},\vec y|_{i})\|_2^2
\]
where $(\vec x|_{i},\vec y|_{i})$ is the restriction of $(\vec x, \vec y)$ to the coordinates corresponding to the $i$th block. In particular, for each $i \in [n/s]$, the pair $(\vec x|_{i},\vec y|_{i})$ forms an $\frac{\eps n}{s}$-$\ell_2$-scaling of $\mat A_i^{\sigma_i}$ to marginals $\vec 1/s$. Hence, for $\eps= C/(ns\sqrt{\ln(s)})$ we recover for each block a scaling with $\ell_2$-error $C/(s^2 \sqrt{\ln(s)})$. For each block this allows us, with probability $\geq 2/3$ over the choice of $\sigma_i$, to compute the Hamming weights of the associated bit strings. Hence, for a suitably large constant $c_0$, using $c_0 \ln(n)$ successful runs of the scaling algorithm with independently drawn choices of the $\sigma_i$'s allows us to compute the Hamming weights of all $n$ bit strings with probability at least $2/3$. The probability that all the runs of the scaling algorithm are successful is at least $(\exp(-C_0 n/\ln(n)))^{c_0\ln(n)} = \exp(-C_0 c_0 n) \geq \frac{3}{2}\exp(-n/100)$, where the last inequality determines the choice of $C_0$. Hence, we compute the Hamming weights of all $n$ bit strings with probability at least $\exp(-n/100)$ and by \cref{thm:basic lb} this requires at least $\Omega(n s)$ quantum queries to the bit strings. 
\end{proof}

\subsection{Lower bound for computing the row marginals} \label{sec: row marginal lb}
In this section we show that computing an $\eps$-$\ell_1$-approximation of the row (or column marginals) of an entrywise-positive $n \times n$ matrix takes $\Omega(n/\eps)$ queries to its entries (for $\eps = \Omega(1/n)$). As a consequence, the same holds for computing an approximation of the gradient of common (convex) potential functions used for matrix scaling -- among which is the potential we use in \cref{sec: upper bound} -- takes as many queries. Although the bound does not imply that testing whether a matrix is $\eps$-$\ell_1$-scaled takes at least $\Omega(n/\eps)$ queries, it gives reasonable evidence that this should be the case.
\begin{theorem} \label{lb: original}
Let $\tau\in[1/n,1/2]$.
Suppose we have a quantum algorithm that, given query access to a positive $n\times n$ matrix $\mat A$ with row-sums $\vec r=(r_1,\ldots,r_n)$ and column-sums $\vec c=(1/n,\ldots,1/n)$, outputs (with probability $\geq \exp(-n/100)$) a vector $\vec{\tilde{r}}\in\mathbb{R}_+^n$ such that $\norm{\vec{\tilde{r}}- \vec r}_1<\tau/100$.
Then this algorithm uses $\Omega(n/\tau)$ queries.
\end{theorem}
\begin{proof}
The strategy is to reduce instances of \cref{thm:basic lb} to the $\ell_1$-approximation of~$\vec r$.
Starting from such an instance, we define a strictly positive $n\times n$ matrix $\mat A$ as we did for the scaling lower bound: first define an $n\times n$ Boolean matrix $\mat B$ by taking $x^{(1)},\ldots,x^{(n/2)}$ as its first $n/2$ rows, and the \emph{negated} versions of those strings as  the last $n/2$ rows (this ensures that each column has exactly $n/2$ 1s and $n/2$ 0s).
Now $\mat A$ is obtained from~$\mat B$ by replacing each 1 by $1.5/n^2$ and each 0 by $0.5/n^2$.
One query to $\mat A$ can be implemented by one query to the $x^{(i)}$s.
Note that each column-sum in $\mat A$ is exactly
\[
  n/2\cdot 0.5/n^2 + n/2\cdot 1.5/n^2=1/n.
\]
For $i\in[n/2]$, the $i$-th row-sum is $r_i=1/n+a_i\tau/n$ (and $r_{n/2+i}=1/n-a_i\tau/n$).
Approximating this row-sum up to additive error $<\tau/n$ by some $\vec{\tilde{r}}_i$ tells us what $a_i$ is: $a_i=\sign(\tilde{r}_i-1/n)$.

Now suppose we have an algorithm for $\ell_1$-approximation of $r$ as in the theorem statement. Using some $T$ queries to $\mat A$ it produces (with probability $\geq \exp(-n/100)$) a vector $\vec{\tilde{r}}$ such that $\norm{\vec{\tilde{r}}-\vec r}_1< \tau/100$.
Then for at least 99\%\ of the $i$ we must have $|\tilde{r}_i-r_i|<\tau/n$.
Defining $\tilde{a}_i=\sign(\tilde{r}_i-1/n)$ for all $i\in[n/2]$, we obtain an $\tilde{a}$ that agrees with $a$ for 99\%\ of the $i$, so by \cref{thm:basic lb}, $T = \Omega(n/\tau)$.
\end{proof}

\section{Quantum box-constrained Newton method for matrix scaling}
\label{sec: upper bound}

In this section, we show how to obtain a quantum speedup based on the box-constrained Newton method for matrix scaling from \cite{cmtv17}, with the main result being \cref{thm:upper bound}, and its consequences for matrix scaling given in \cref{cor: non-positive,cor: positive}.
We first recall some of the concepts that are used in the algorithm, including the definition of second-order robust convex functions, the notion of a $k$-oracle, and a theorem regarding efficient (classical) implementation of a $k$-oracle for the class of symmetric diagonally-dominant matrices with non-positive off-diagonal entries.
We then show that for a second-order robust function $g\colon \R^n \to \R$ and a given $\vec x \in \R^n$ such that the sublevel set $\{\vec x' : g(\vec x') \leq g(\vec x) \}$ is bounded, one can use a $k$-oracle and approximations to the gradient and Hessian of $g$ to find a vector $\vec x'$ such that the potential gap $g(\vec x') - g(\vec x^*)$ is smaller than $g(\vec x) - g(\vec x^*)$ where $\vec x^*$ is a minimizer of $g$.
This result extends \cite[Thm.~3.4]{cmtv17} to a setting where one can only obtain rough approximations of the gradient and Hessian of $g$.
We then show that this applies to a regularized version $\tilde f$ of the potential $f$ discussed in the introduction; to approximate the Hessian of $\tilde f$, we use a quantum algorithm for graph sparsification, whereas we approximate the gradient of $\tilde f$ using quantum approximate summing.
One challenge is that the quality of the gradient approximation is directly related to the $1$-norm of the matrix $\A x y$, so we must control this throughout the algorithm, which we achieve by manually shifting $\vec x$ when the norm becomes too large, and showing that this does not increase the regularized potential under suitable circumstances.

\subsection{Minimizing second-order robust convex functions} \label{sec: minimizing second-order robust}

In what follows we will minimize a convex function (potential) that satisfies a certain regularity condition: its Hessian can be approximated well on an infinity-norm ball. 
\begin{definition}[{\cite[Def.~3.1]{cmtv17}}]
  A convex function $g\colon \R^n \to \R$ is called \emph{second-order robust} with respect to $\ell_\infty$ if for any $\vec x, \vec y \in \R^n$ with $\norm{\vec x - \vec y}_\infty \leq 1$, 
  \begin{equation*}
    \frac{1}{e^2} \hess g(\vec x) \preceq \hess g(\vec y) \preceq e^2 \hess g(\vec x).
  \end{equation*}
\end{definition}

This implies that the local quadratic approximation to $g$ has a good quality on a small $\ell_\infty$-norm ball. It is therefore natural to consider the problem of minimizing a convex quadratic function over an $\ell_\infty$-norm ball. We will use the following notion. 

\begin{definition}[$k$-oracle]
  \label{def:k-oracle}
  An algorithm $\mathcal A$ is called a \emph{$k$-oracle} for a class of matrices $\mathcal M \subseteq \R^{n \times n}$ if for input $(\mat H, \vec b)$ with $\mat H \in \mathcal M$, $\vec b \in \R^n$, it returns a vector $\vec x \in \R^n$ such that $\norm{x}_\infty \leq k$ and 
  \begin{equation}
    \label{eq:localopt}
    \frac12 \vec x^T \mat H \vec x + \ip{\vec b}{\vec x} \leq \frac12 \cdot \min_{\norm{\vec z}_\infty \leq 1} (\frac12 \vec z^T \mat H \vec z + \ip{\vec b}{\vec z}).
  \end{equation}
\end{definition}
\begin{definition}[SDD matrix]
  \label{def:SDD matrix}
  A matrix $\mat A \in \R^{n \times n}$ is called \emph{symmetric diagonally-dominant} if it is symmetric, and for every $i \in [n]$, one has $A_{ii} \geq \sum_{j \neq i} \abs{A_{ij}}$.
\end{definition}

In \cite{cmtv17} it is shown how to efficiently implement an $\bigO{\log(n)}$-oracle for the class of SDD matrices $\mat H$ whose off-diagonal entries are non-positive. Their algorithm uses an efficient construction of a \emph{vertex sparsifier chain} of $\mat H$ due to~\cite{lee2015sparsified,DBLP:conf/stoc/KyngLPSS16}. 
\begin{theorem}[{\cite[Thm.~5.11]{cmtv17}}] \label{thm: log n oracle}
Given a classical description of an SDD matrix $\mat H \in \R^{n \times n}$ with $\bigOt{m}$ non-zero entries, such that $ H_{i,j} \leq 0$ for $i \neq j$, and a classical vector $\vec b \in \R^n$, we can find in time $\tilde O(m)$ a vector $\vec x \in \R^n$ such that $\norm{\vec x}_\infty = O(\log n)$ and
\[
\frac12 \vec x^T\mat H \vec x + \ip{\vec b}{\vec x} \leq \frac12 \cdot \min_{\norm{\vec z}_\infty \leq 1} (\frac12 \vec z^T\mat H \vec z + \ip{\vec b}{\vec z}).
\]
\end{theorem}

A $k$-oracle $\mathcal A$ gives rise to an iterative method for minimizing a second-order robust function $g$: starting from $x_0 \in \R^n$, we define a sequence $\vec x^{(0)}, \vec x^{(1)}, \vec x^{(2)},\ldots$ by
 \begin{equation*}
    \vec x^{(i+1)} = \vec x^{(i)} + \frac1k \Delta_i, \quad \Delta_i = \mathcal A\left(\frac{e^2}{k^2} \mat H_i, \frac{1}{k} \vec b_i\right)
  \end{equation*}
  where $\mat H_i$ is an approximate Hessian at $\vec x^{(i)}$, and $\vec b_i$ is an approximate gradient at $\vec x^{(i)}$. The following theorem, which is an adaptation of \cite[Thm.~3.4]{cmtv17}, upper bounds the progress made in each iteration.

\begin{theorem} \label{thm: second order robust update}
    Let $g\colon \R^n \to \R$ be a second-order robust function with respect to $\ell_\infty$, let $\vec x \in \R^n$ be a starting point, and suppose $\vec x^*$ is a minimizer of $g$.
    Assume that we are given 
    \begin{enumerate}[label=(\arabic*)]
        \item a vector $\vec b \in \R^n$ such that
          \begin{equation*}
            \norm{\vec b - \grad g(\vec x)}_1 \leq \delta,
          \end{equation*}
        \item two SDD matrices $\mat H_m$ and $\mat H_a$ with non-positive off-diagonal entries, such that there exists $\delta_a \geq 0$ and symmetric $\mat H_m'$ and $\mat H_a'$ satisfying $\hess g(\vec x) = \mat H_m' + \mat H_a'$ and
        \begin{equation*}
             \frac23 \mat H_m \preceq \mat H_m' \preceq \frac43 \mat H_m, \quad \norm{\mat H_a - \mat H_a'}_1 \leq \delta_a.
        \end{equation*}
    \end{enumerate}
    Let $k=\bigO{\log n}$ be such that there exists a $k$-oracle $\mathcal A$ for the class of SDD-matrices with non-positive off-diagonal entries (cf.~\cref{thm: log n oracle}).
    Then for $\mat H = \mat H_m + \mat H_a$ and $\Delta = \mathcal A\left(\frac{4e^2}{3k^2} \mat H, \frac{1}{k} \vec b\right)$, the vector $\vec x' = \vec x + \frac1k \Delta$ satisfies
    \begin{equation*}
      g(\vec x') - g(\vec x^*) \leq \left(1 - \frac{1}{4e^4 \max(k R_\infty, 1)}\right) (g(\vec x) - g(\vec x^*)) + \frac{e^2 \delta_a}{k^2} + \frac32\delta,
    \end{equation*}
    where $R_\infty$ is the $\ell_\infty$-radius of the sublevel set $\{\vec x' : g(\vec x') \leq g(\vec x)\}$ about $\vec x$.
\end{theorem}
Before giving the proof, we introduce the following notation. For a symmetric matrix $\mat H$ and $\vec b, \vec z \in \R^n$, we denote
\begin{equation*}
  Q(\mat H, \vec b, \vec z) = \ip{\vec b}{\vec z} + \frac12 \vec z^T \mat H \vec z.
\end{equation*}
We will use the following easily-verified properties of $Q$ repeatedly. 
\begin{lemma}
    For symmetric matrices $\mat H, \mat H'$ and vectors $\vec b, \vec b', \vec z$, we have the following estimates:
    \begin{enumerate}
        \item If $\mat H \preceq \mat H'$, then $Q(\mat H, \vec b, \vec z) \leq Q(\mat H', \vec b, \vec z)$.
        \item If $\norm{\mat H - \mat H'}_1 \leq \delta_a$, then 
        \[
            \abs*{Q(\mat H, \vec b, \vec z) - Q(\mat H', \vec b, \vec z)} \leq \frac12 \delta_a \norm{\vec z}_\infty^2.
        \]
        \item We have
        \[
                \abs*{Q(\mat H, \vec b, \vec z) - Q(\mat H, \vec b', \vec z)} = \abs*{\ip{\vec b - \vec b'}{\vec z}} \leq \norm{\vec b - \vec b'}_1 \norm{\vec z}_\infty.
        \]
    \end{enumerate}
\end{lemma}

\begin{proof}[Proof of~\cref{thm: second order robust update}]
  We follow the proof of \cite[Thm.~3.4]{cmtv17}, and use their implementation of a $k$-oracle $\mathcal A$ for $k=\bigO{\log n}$, as detailed in \cref{thm: log n oracle}.
  That is, $\mathcal A$ takes as input an SDD matrix $\mat H$ with $\bigOt{m}$ non-zero entries (off-diagonal entries $\leq 0$) and a vector $\vec b$, and outputs a vector $\vec z$ such that $\norm{\vec z}_\infty \leq k$ and
  \begin{equation*}
    Q(\mat H, \vec b, \vec z) \leq \frac{1}{2} \inf_{\norm{\vec z'}_\infty \leq 1} Q(\mat H, \vec b, \vec z').
  \end{equation*}
Then for   
  \begin{equation*}
    \vec x' = \vec x + \frac1k \Delta, \quad \Delta = \mathcal A\left(\frac{4e^2}{3k^2} \mat H, \frac{1}{k} \vec b\right)
  \end{equation*}
we have 
  \begin{align*}
    Q\left(\frac{4e^2}{3} \mat H, \vec b, \frac{1}{k} \Delta\right) & = 
      Q\left(\frac{4e^2}{3k^2} \mat H, \frac{1}{k} \vec b, \Delta \right) \\
      & \leq \frac{1}{2} \inf_{\norm{z}_\infty \leq 1} Q\left(\frac{4e^2}{3k^2} \mat H, \frac{1}{k} \vec b, \vec z \right) \\
      & = \frac{1}{2} \inf_{\norm{z}_\infty \leq 1} Q\left(\frac{4e^2}{3} \mat H, \vec b, \vec z/k \right)\\
      & = \frac{1}{2} \inf_{\norm{z}_\infty \leq \frac{1}{k}} Q\left(\frac{4e^2}{3} \mat H, \vec b, \vec z \right).
  \end{align*}
  Note that the second-order robustness of $g$ implies that for $\vec{\tilde x} \in \R^n$ with $\norm{\vec x - \vec{\tilde x}}_\infty \leq 1$, we have quadratic lower and upper bounds
  \begin{equation}
    \label{eq:g lower and upper bound on box}
      Q\left(\frac{1}{e^2} \hess g(\vec x), \grad g(\vec x), \vec{\tilde x} - \vec x\right) \leq g(\vec{\tilde x}) - g(\vec x) \leq Q\left(e^2 \hess g(\vec x), \grad g(\vec x), \vec{\tilde x} - \vec x\right).
  \end{equation}
  The remainder of the proof is structured as follows. We first compare quadratics involving $\hess g(\vec x)$ and $\nabla g(\vec x)$ to quadratics involving the approximations $\mat H$ and $\vec b$ in \cref{eq:long local inequality chain pt1,eq:long local inequality chain pt2}. Using these estimates we then obtain a local progress bound over an $\ell_\infty$-ball of radius $1/k$, see \cref{eq:intermediateestimate}. Finally, we convert this local bound into a more global estimate.
  
  The properties of the approximate Hessian and gradient guarantee that
  \begin{equation}
    \label{eq:long local inequality chain pt1}
    \begin{split}
      & Q\left(e^2 \hess g(\vec x), \grad g(\vec x), \vec{\tilde x} - \vec x\right) \\
      & \leq Q\left(e^2 \hess g(\vec x), \vec b, \vec{\tilde x} - \vec x\right) + \delta \\
      & = Q\left(e^2 \mat H_m', \vec b, \vec{\tilde x} - \vec x\right) + Q\left(e^2 \mat H_a', \vec b, \vec{\tilde x} - \vec x\right) - \ip{\vec b}{\vec {\tilde x} - \vec x} + \delta \\
      & \leq Q\left(\frac{4e^2}{3} \mat H_m, \vec b, \vec{\tilde x} - \vec x\right) + Q\left(e^2 \mat H_a, \vec b, \vec{\tilde x} - \vec x\right) + \frac{e^2}{2} \delta_a \norm{\vec{\tilde x} - \vec x}_\infty^2 - \ip{\vec b}{\vec {\tilde x} - \vec x} + \delta \\
      & \leq Q\left(\frac{4e^2}{3} \mat H_m, \vec b, \vec{\tilde x} - \vec x\right) + Q\left( \frac{4e^2}{3} \mat H_a, \vec b, \vec{\tilde x} - \vec x\right) + \frac{e^2}{2} \delta_a \norm{\vec{\tilde x} - \vec x}_\infty^2 - \ip{\vec b}{\vec {\tilde x} - \vec x} + \delta \\
      & = Q\left(\frac{4e^2}{3} \mat H, \vec b, \vec{\tilde x} - \vec x\right) + \frac{e^2}{2} \delta_a \norm{\vec{\tilde x} - \vec x}_\infty^2 + \delta.
    \end{split}
  \end{equation}
  Furthermore, we also have the upper bound
  \begin{equation}
      \label{eq:long local inequality chain pt2}
      \begin{split}
          & Q\left(\frac{4e^2}{3} \mat H, \vec b, \vec{\tilde x} - \vec x\right) \\
          & = Q\left(\frac{4e^2}{3} \mat H_m, \vec b, \vec{\tilde x} - \vec x\right) + Q\left( \frac{4e^2}{3} \mat H_a, \vec b, \vec{\tilde x} - \vec x\right) - \ip{\vec b}{\vec {\tilde x} - \vec x} \\
          & \leq Q\left(2e^2 \mat H_m', \vec b, \vec{\tilde x} - \vec x\right) + Q\left( 2 e^2 \mat H_a, \vec b, \vec{\tilde x} - \vec x\right) - \ip{\vec b}{\vec {\tilde x} - \vec x}  \\
          & \leq Q\left(2e^2 \mat H_m', \vec b, \vec{\tilde x} - \vec x\right) + Q\left( 2 e^2 \mat H_a', \vec b, \vec{\tilde x} - \vec x\right) + e^2 \delta_a \norm{\vec{\tilde x} - \vec x}_\infty^2 - \ip{\vec b}{\vec {\tilde x} - \vec x}  \\
          & \leq Q\left(2 e^2 \mat H_m', \vec b, \vec{\tilde x} - \vec x\right) + Q\left( 2 e^2 \mat H_a', \vec b, \vec{\tilde x} - \vec x\right) + e^2 \delta_a \norm{\vec{\tilde x} - \vec x}_\infty^2 - \ip{\vec b}{\vec {\tilde x} - \vec x}  \\
          & = Q\left(2 e^2 \hess g(\vec x), \vec b, \vec{\tilde x} - \vec x\right) + e^2 \delta_a \norm{\vec{\tilde x} - \vec x}_\infty^2  \\
          & \leq Q\left(2 e^2 \hess g(\vec x), \grad g(\vec x), \vec{\tilde x} - \vec x\right) + e^2 \delta_a \norm{\vec{\tilde x} - \vec x}_\infty^2  + \delta.
      \end{split}
  \end{equation}
  Let $\vec v_L$ and $\vec v_U$ be the minimizers of quadratics over the $\ell_\infty$-ball of radius $1/k$: 
  \begin{equation*}
      \vec v_L = \argmin_{\norm{\vec v}_\infty \leq 1/k} Q(\frac{1}{e^2} \hess g(\vec x), \grad g(\vec x), \vec v),
      \quad
      \vec v_U = \argmin_{\norm{\vec v}_\infty \leq 1/k} Q(2 e^2 \hess g(\vec x), \grad g(\vec x), \vec v).
  \end{equation*}
  Then by the guarantees of the $k$-oracle, we have
  \begin{align*}
      Q\left(\frac{4e^2}{3} \mat H, \vec b, \frac{1}{k} \Delta \right) & \leq \frac{1}{2} \inf_{\norm{\vec v}_\infty \leq 1/k} Q\left(\frac{4e^2}{3} \mat H, \vec b, \vec v \right) \\
                                                                          & \leq \frac{1}{2} \inf_{\norm{\vec v}_\infty \leq 1/k} (Q\left(2 e^2 \hess g(\vec x), \grad g(\vec x), \vec v \right) + e^2 \delta_a \norm{\vec v}_\infty^2 + \delta) \\
                                                                          & \leq \frac{1}{2} Q\left(2 e^2 \hess g(\vec x), \grad g(\vec x), \vec v_U \right) + \frac{e^2\delta_a}{2 k^2} + \frac12 \delta,
  \end{align*}
  where the second inequality uses \cref{eq:long local inequality chain pt2}, and the norm bounds $\norm{\vec v}_\infty \leq 1/k \leq 1$ (to apply the inequality).
  Using the quadratic upper bound from \cref{eq:g lower and upper bound on box} on $g(\vec x + \frac{1}{k} \Delta) - g(\vec x)$ and \cref{eq:long local inequality chain pt1}, this yields
  \begin{align*}
      g(\vec x + \frac{1}{k} \Delta) - g(\vec x) & \leq Q(e^2 \hess g(\vec x), \grad g(\vec x), \frac{1}{k} \Delta) \\
      & \leq Q\left(\frac{4e^2}{3} \mat H, \vec b, \frac{1}{k} \Delta \right) + \frac{e^2}{2} \delta_a + \delta\\
        & \leq \frac{1}{2} Q\left(2 e^2 \hess g(\vec x), \grad g(\vec x), \vec v_U \right) + \frac{e^2 \delta_a}{k^2} + \frac32 \delta,
  \end{align*}
  We can then further upper bound this using 
  \begin{align*}
      Q\left(2 e^2 \hess g(\vec x), \grad g(\vec x), \vec v_U \right) & \leq Q\left(2 e^2 \hess g(\vec x), \grad g(\vec x), \frac{\vec v_L}{2 e^4} \right) \\
      & = \frac{1}{2 e^4} Q\left(\frac{1}{e^2} \hess g(\vec x), \grad g(\vec x), \vec v_L\right) 
  \end{align*}
  where the inequality uses that $\vec v_U = \argmin_{\|v\|_\infty \leq 1/k} Q(2 e^2 \hess g(\vec x), \grad g(\vec x), v)$ and $\|\vec v_L\|_\infty \leq 1/k$. 
  Collecting estimates, we obtain
  \begin{equation}
      \label{eq:intermediateestimate}
      g(\vec x + \frac{1}{k} \Delta) - g(\vec x) \leq \frac{1}{4 e^4} Q\left(\frac{1}{e^2} \hess g(\vec x), \grad g(\vec x), \vec v_L\right) + \frac{e^2 \delta_a}{k^2} + \frac32 \delta.
  \end{equation}
  
  We now convert this to a more global estimate.
  Let $\vec x^*$ be a global minimizer of $g$. Set $\vec y = \vec x + \frac{1}{\max(k R_\infty, 1)} (\vec x^* - \vec x)$, so that $\norm{\vec y - \vec x}_\infty \leq \frac{1}{k}$.
  For the lower bound 
  \begin{equation*}
    g_L(\vec{\tilde x}) = g(\vec x) + Q(\frac{1}{e^2} \hess g(\vec x), \grad g(\vec x), \vec{\tilde x} - \vec x)
  \end{equation*}
  on $g(\vec{\tilde x})$ we see that $g_L(\vec x + \vec v_L) \leq g_L(\vec y) \leq g(\vec y)$ since $\vec x + \vec v_L$ minimizes $g_L \leq g$ over the $\ell_\infty$-ball of radius $1/k$ around $\vec x$.
  By convexity of $g$ we get
  \begin{align*}
      g(\vec y) & = g(\vec x + \frac{1}{\max(k R_\infty, 1)} (\vec x^* - \vec x)) \\
                & \leq (1 - \frac{1}{\max(k R_\infty, 1)}) g(\vec x) + \frac{1}{\max(k R_\infty, 1)} g(\vec x^*)
  \end{align*}
  so
  \begin{equation*}
      g(\vec x) - g_L(\vec x + \vec v_L) \geq g(\vec x) - g(\vec y)  \geq \frac{1}{\max(k R_\infty, 1)} (g(\vec x) - g(\vec x^*)).
  \end{equation*}
  Using this estimate in \cref{eq:intermediateestimate}, this gives
  \begin{equation} \notag
      g(\vec x) - g(\vec x + \frac{1}{k} \Delta) \geq \frac{1}{4 e^4 \max(k R_\infty, 1)} (g(\vec x) - g(\vec x^*)) - (\frac{e^2 \delta_a}{k^2} + \frac32 \delta),
  \end{equation}
  which after rearranging and rewriting $\vec x' = \vec x + \frac1k \Delta$ reads
  \begin{equation*}
    g(\vec x') - g(\vec x^*) \leq \left(1 - \frac{1}{4e^4 \max(k R_\infty, 1)}\right) (g(\vec x) - g(\vec x^*)) + \frac{e^2 \delta_a}{k^2} + \frac32 \delta. \qedhere
  \end{equation*}
\end{proof}

\subsection{A second-order robust potential for matrix scaling and its properties}

Given a sparse matrix $\mat A \in \R_{\geq 0}^{n \times n}$, a desired error $\eps > 0$, and some number $B > 0$, we consider the regularized potential function $\tilde f(x,y)$ given by
\begin{equation*}
  \tilde f(\vec x, \vec y) = f(\vec x, \vec y) + \frac{\eps^2}{n e^B} \left( \sum_i (e^{x_i} + e^{-x_i}) + \sum_j (e^{y_j} + e^{-y_j}) \right),
\end{equation*}
where $f$ is the commonly-used potential function from \cref{def: f}. In \cite{cmtv17}, the same regularization term is used, but with a different weight (since they aim for $\ell_2$-scaling and we aim for $\ell_1$-scaling). The following is then an adaptation of \cite[Lem.~4.10]{cmtv17}.
\begin{lemma}
\label{lem:basic regularized properties}
Assume $\mat A$ is asymptotically scalable, with $\norm{\mat A}_1 \leq 1$, and $\mu > 0$ its smallest non-zero entry.
Let $B > 0$ and $\eps > 0$ be given. Then the regularized potential $\tilde f$ satisfies the following properties:
\begin{enumerate}
    \item $\tilde f$ is second-order robust with respect to $\ell_\infty$, and its Hessian is SDD;
    \item we have $f(\vec z) \leq \tilde f(\vec z)$ for any $\vec z = (\vec x,\vec y)$,
    \item for all $\vec z$ such that $\tilde f(\vec z) \leq \tilde f(\vec 0)$, we have $\norm{\vec z}_\infty \leq B + \ln(4n + (n \ln(1/\mu) / \eps^2))$, and 
    \item for any $\vec z_\eps$ such that $f(\vec z_\eps) \leq f^* + \eps^2$ and $\norm{\vec z_\eps}_\infty \leq B$, one has $\tilde f(\vec z_\eps) \leq f^* + 5 \eps^2$. 
    In particular, if such a $\vec z_\eps$ exists, then $\abs{f^* - \tilde f^*} \leq 5 \eps^2$.
\end{enumerate}
\end{lemma}
\begin{proof}
The first point is easy to verify, as is the second point (the regularization term is always positive).
For the third point, suppose we have a $\vec z$ such that $\tilde f(\vec z) \leq \tilde f(0)$. Then
\begin{equation} \label{eq: sum exp bound}
    \frac{\eps^2}{ne^B} \left( \sum_i (e^{x_i} + e^{-x_i}) + \sum_j (e^{y_j} + e^{-y_j}) \right) \leq f(\vec 0) - f(\vec z) + \frac{\eps^2}{ne^B} \cdot 4n \leq \ln(1/\mu) + \frac{4 \eps^2}{e^B}.
\end{equation}
where the last inequality follows from the potential bound $f(\vec 0) - f^* \leq \ln(1/\mu)$ (which depends on $\norm{\mat A}_1 \leq 1$; in general the upper bound is $\norm{\mat A}_1 - 1 + \ln(1/\mu)$).
Since each of the regularization terms is positive, we may restrict ourselves to a single term and see that 
\begin{equation*}
    e^{x_i} + e^{-x_i} \leq \frac{e^B n \ln(1/\mu)}{\eps^2} + 4n,
\end{equation*}
from which we may deduce
\begin{equation*}
    \abs{x_i} \leq \ln\left(\frac{e^B n \ln(1/\mu)}{\eps^2} + 4n \right) = B + \ln\left(\frac{n \ln(1/\mu)}{\eps^2} + \frac{4n}{e^B} \right) \leq B + \ln\left(\frac{n \ln(1/\mu)}{\eps^2} + 4n \right),
\end{equation*}
where the last inequality uses $e^B \geq 1$ (recall $B > 0$). The same upper bound holds for $\abs{y_j}$.

For the last point, note that if $\vec z_\eps = (\vec x, \vec y)$, then $e^{x_i} + e^{-x_i} \leq 2 e^B$ and similarly for $y$, so 
\begin{equation*}
    \tilde f(\vec z_\eps) \leq f(\vec z_\eps) + \frac{\eps^2}{n e^B} \cdot 4n e^B = f(\vec z_\eps) + 4 \eps^2 \leq f^* + 5 \eps^2. 
\end{equation*}
If such a $\vec z_\eps$ exists, then 
\begin{equation*}
    f^* \leq \tilde f^* \leq \tilde f(\vec z_\eps) \leq f^* + 5 \eps^2. \qedhere
\end{equation*}
\end{proof}

In order to use \cref{thm: second order robust update} to minimize $f$, we need to show how to approximate both the gradient and Hessian of $\tilde f$. We first consider the Hessian of $\tilde f$, which can be written as the sum of the Hessian of $f$ and the Hessian of the regularizer $\tilde f -f$. We have 
\begin{equation}
\begin{split}
\nabla^2f(\vec x,\vec y) &= \begin{bmatrix} \diag(\vec r(\A  x y)) & \A x y \\ {\A x y}^T & \diag(\vec c(\A  x y)) \end{bmatrix}, \\
\nabla^2 (\tilde f-f)(\vec x,\vec y) &= \frac{\eps^2}{n e^B}\begin{bmatrix} \diag(e^{\vec x} + e^{-\vec x})  & \mat 0 \\ \mat 0 & \diag(e^{\vec y} + e^{-\vec y}) \end{bmatrix}.
\end{split}
\end{equation}
Note that computing $\nabla^2 \tilde f(\vec x, \vec y)$ up to high precision can be done using $\bigOt{m}$ classical queries to $\mat A$, $\vec x$, and $\vec y$. Below we show how to obtain a sparse approximation of $\nabla^2 \tilde f(\vec x, \vec y)$ using only $\bigOt{\sqrt{mn}}$ quantum queries. We will do so in the sense of condition (2) of \cref{thm: second order robust update} where we take $\mat H_m'$ to be a (high-precision) additive approximation of $\nabla^2f(\vec x,\vec y)$, and $\mat H_a' = \nabla^2 \tilde f(\vec x,\vec y) - \mat H_m'$.

We first obtain a multiplicative spectral approximation of (a high-precision additive approximation of) $\nabla^2f(\vec x, \vec y)$. In order to do so we use its structure: it is similar to a Laplacian matrix. This allows us to use the recent quantum Laplacian sparsifier of Apers and de Wolf~\cite{Apers2020QLaplacian}.
\begin{lemma}
  \label{lem:approximating potential Hessian}
  Given quantum query access to $\vec x, \vec y$ and sparse quantum query access to $\mat A$, such that $\norm{\A x y}_1 \leq C$, we can compute an SDD matrix $\mat H_m$ with $\bigOt{n}$ non-zero entries, each off-diagonal entry non-negative, such that there exist symmetric $\mat H_m'$ and $\mat H_{a,f}'$ satisfying $\mat H_m' + \mat H_{a,f}' = \hess f(\vec x, \vec y)$, and 
  \[
    0.9 \mat H_m \preceq \mat H_m' \preceq 1.1 \mat H_m, \quad \norm{\mat H_{a,f}'}_1 \leq \delta_a,
  \]
  in time $\bigOt{\sqrt{mn} \polylog(C/\delta_a)}$.
\end{lemma}
\begin{proof}
  The key observation is that $\hess f(\vec x, \vec y)$ satisfies
  \[
    \mat H = \begin{bmatrix} \mat I & \mat 0 \\ \mat 0 & -\mat I \end{bmatrix}\nabla^2 f(\vec x, \vec y) \begin{bmatrix} \mat I & \mat 0 \\ \mat 0 & -\mat I \end{bmatrix} = \begin{bmatrix} \diag(\vec r(\A  x y)) & -\A x y \\ -{\A x y}^T & \diag(\vec c(\A  x y)) \end{bmatrix},
  \]
  which is the Laplacian of the bipartite graph whose bipartite adjacency matrix is given by $\A x y$.
  Any off-diagonal entry of $\mat H$ can be computed with additive error $\delta_a/2(2m + 2n)$ using a single query to $\mat A$, to $\vec x$ and to $\vec y$: the $(i,j)$-th entry of $\A x y$ is $A_{ij} e^{x_i + y_j}$, which is at most $C$ (since $\norm{\A x y}_1 \leq C$ by assumption), so we can compute $e^{x_i + y_j}$ to sufficient precision ($\lceil\log_2(C/\mu)\rceil + O(1)$ leading bits and $\lceil \log_2(2(2m+2n)/\delta_a) \rceil + O(1)$ trailing bits) and multiply it with $A_{ij}$.
  We can do this in such a way that if $A_{ij} = 0$, then the resulting entry is $0$, and such that the approximation of $A_{ij} e^{x_i + y_j}$ is always non-negative.
  
  Let $\mat H'$ be the matrix whose off-diagonal entries are given by these approximations of the corresponding entries of $\mat H$, and whose diagonal entries are such that $\mat H'$ is Laplacian.
  Then $\norm{\mat H' - \mat H}_1 \leq \delta_a$ by the chosen precision for the additive approximation.
  Furthermore, as described before, a single query to off-diagonal entries of $\mat H'$ can be implemented using a single query to $\mat A$, $\vec x$ and $\vec y$.
  Theorem~1 of \cite{Apers2020QLaplacian} gives a quantum algorithm that uses $\bigOt{\sqrt{mn}}$ queries to the off-diagonal entries of $\mat H'$ and outputs a $0.1$-spectral sparsification $\mat {\tilde H}$ of $\mat H'$ that has $\bigOt{n}$ non-zero entries.
  Note that every non-zero entry of $\mat {\tilde H}$ was already non-zero in $\mat H'$ because it is the Laplacian of a reweighted subgraph of the graph described by $\mat H'$; hence any non-zero off-diagonal entry in $\mat {\tilde H}$ is contained in either the upper right or lower left $n \times n$ block, and each such entry is non-positive.
  Then the matrix $\mat H_m = \diag(I, -I) \mat{\tilde H} \diag(I, -I)$ satisfies the conclusion in the lemma, with $\mat H_m' = \diag(I, -I) \mat H' \diag(I, -I)$ and $\mat H_{a,f}' = \diag(I, -I) (\mat H' - \mat H) \diag(I, -I)$.
\end{proof}

We now show how to compute an additive approximation of the Hessian of the regularization term in $\tilde f$.
\begin{lemma}
\label{lem:approximating regularization Hessian}
Given quantum query access to $\vec x, \vec y$ with $\norm{\vec x}_\infty, \norm{\vec y}_\infty \leq B + \ln(4n + (n \ln(1/\mu) / \eps^2))$, we can compute a non-negative diagonal matrix $\mat H_{a, \tilde f}$ that satisfies $\norm{\mat H_{a, \tilde f} - \nabla^2(\tilde f- f)(\vec x, \vec y)}_1 \leq \delta_a$, in time $\bigOt{n\log(1/\delta_a \mu) \polylog(\eps)}$.
\end{lemma}
\begin{proof}
Recall that $\nabla^2 (\tilde f-f)(\vec x,\vec y)$ is a diagonal matrix whose entries are of the form $\frac{\eps^2}{n e^B}( e^{x_i} + e^{-x_i})$ or $\frac{\eps^2}{n e^B}( e^{y_i} + e^{-y_i})$. Note that by assumption on the $\ell_\infty$-norms of $\vec x$ and $\vec y$, all diagonal entries are upper bounded by
\[
2 \frac{\eps^2}{n e^B}e^{B + \ln(4n + (n \ln(1/\mu) / \eps^2))} = 2 \frac{\eps^2}{n}e^{\ln(4n + (n \ln(1/\mu) / \eps^2))} = 2 \frac{\eps^2}{n} (4n + (n \ln(1/\mu) / \eps^2)) = 8\eps^2 +2\ln(1/\mu).
\]
Hence, it suffices to compute each diagonal entry using $\lceil \log_2(8\eps^2 +2\ln(1/\mu))\rceil$ leading bits and $\lceil \log_2(1/n\delta_a)\rceil$ trailing bits. We can do so efficiently by using the identity 
\[
\frac{\eps^2}{n e^B}( e^{x_i} + e^{-x_i}) = \frac{\eps^2}{n}(e^{x_i-B} + e^{-x_i-B})
\]
and the analogous one for $y_i$.
\end{proof}

\begin{theorem}
  \label{thm:approximating regularized potential Hessian}
  Given quantum query access to $\vec x, \vec y$ with $\norm{\vec x}_\infty, \norm{\vec y}_\infty \leq B + \ln(4n + (n \ln(1/\mu) / \eps^2))$, and sparse quantum query access to $\mat A$, if $\norm{\A x y}_1 \leq C$, then we can compute (classical descriptions of) an SDD matrix $\mat H_m$ with $\bigOt{n}$ non-zero entries, with all of the off-diagonal entries non-negative, and a non-negative diagonal matrix $\mat H_a$ such that there exist symmetric $\mat H_m'$, $\mat H_a'$ with $\mat H_m' + \mat H_a' = \hess \tilde f(\vec x, \vec y)$ and
  \[
    0.9 \mat H_m \preceq \mat H_m' \preceq 1.1 \mat H_m, \quad \norm{\mat H_a - \mat H_a'}_1 \leq \delta_a
  \]
  in quantum time $\bigOt{\sqrt{mn} \polylog(C/ \mu\delta_a)}$.
\end{theorem}
\begin{proof}
  Let $\mat H_m$ be the matrix obtained from \cref{lem:approximating potential Hessian}, and let $\mat H_a$ be the matrix $\mat H_{a,\tilde f}$ obtained from \cref{lem:approximating regularization Hessian}. Then $\mat H$ satisfies the desired properties, with $\mat H_m'$ as in \cref{lem:approximating potential Hessian}, and $\mat H_a' = \mat H_{a,f}' + \hess (\tilde f - f)(\vec x, \vec y)$ with $\mat H_{a,f}'$ as in \cref{lem:approximating potential Hessian}.
\end{proof}

In order to obtain a good approximation of the gradient of $\tilde f$, which is given by
\[
    \grad \tilde f(\vec x, \vec y) = \begin{bmatrix}
        r_1(\A x y) - r_1 \\
        \vdots \\
        r_n(\A x y) - r_n \\
        c_1(\A x y) - c_1 \\
        \vdots \\
        c_n(\A x y) - c_n
    \end{bmatrix} 
    + 
    \frac{\eps^2}{n e^B}
    \begin{bmatrix}
        e^{x_1} - e^{-x_1} \\
        \vdots \\
        e^{x_n} - e^{-x_n} \\
        e^{y_1} - e^{-y_1} \\
        \vdots \\
        e^{y_n} - e^{-y_n}
    \end{bmatrix}\!,
\]
we can use similar techniques as the prior work on quantum algorithms for matrix scaling~\cite{qscalingICALP}. For computing the $i$-th row marginal, these are based on a careful implementation of amplitude estimation on the unitary that prepares states that are approximately of the form 
\[
\sum_{j} \ket{0}\sqrt{A_{ij} e^{x_i+y_j}}\ket{j} + \ket{1}\sqrt{1-A_{ij} e^{x_i+y_j}}\ket{j},
\]
assuming that the $i$-th row of $\A x y$ is properly normalized. The output is an estimate of the $i$-th row marginal with multiplicative error $1 \pm \delta$, which translates into additive error $\delta \cdot r_i(\A x y)$; we refer to \cite[Thm.~4.5 (arXiv)]{qscalingICALP} for a more precise statement. The part of the gradient coming from the regularization term is dealt with similarly as in \cref{lem:approximating regularization Hessian}.
\begin{lemma}
  \label{lem:approxgrad}
  Given quantum query access to $\vec x, \vec y$ and sparse quantum query access to $\mat A$, if $\norm{\A x y}_1 \leq C$, we can find a classical description of a vector $\vec b \in \R^n$ such that
  \begin{equation*}
    \norm{\vec b - \grad \tilde f(\vec x, \vec y)}_1 \leq \delta \cdot C 
  \end{equation*}
  in quantum time $\bigOt{\sqrt{mn}/\delta \cdot \polylog(C/\mu)}$.
\end{lemma}

\paragraph{Controlling the $1$-norm of $\A  x y$: a Sinkhorn step}

The following lemma and corollary help us ensure that throughout the algorithm, $\norm{\A x y}_1$ is bounded above by a constant; if $\norm{\A x y}_1$ is too large, we can change the overall scaling of the matrix and decrease the regularized potential (so in particular, we stay in the sublevel set of the regularized potential).
\begin{lemma}
Let $\vec x, \vec y$ be such that $\tilde f(\vec x, \vec y) \leq \tilde f(\vec 0, \vec 0)$, and assume $\norm{\A x y}_1 \geq C'$ where $C'>1$. Let $\vec x' = \vec x - \ln(\gamma) \vec 1$ where $1 \leq \gamma \leq C'$. Then 
\[
\tilde f(\vec x',\vec y) - \tilde f(\vec x, \vec y) \leq (\frac1\gamma - 1) C' + \ln(\gamma) + (\gamma-1)\left(\ln(1/\mu)+\frac{4\eps^2}{e^B}\right)  
\]
\end{lemma}
\begin{proof}
We have 
\begin{align*}
    \tilde f(\vec x',\vec y) - \tilde f(\vec x,\vec y) &= (\frac1\gamma - 1) \| \A x y \|_1 + \langle \vec r, \ln(\gamma)\vec 1\rangle  +  \frac{\eps^2}{n e^B} \left(\sum_i (e^{x_i-\ln(\gamma)}-e^{x_i} + e^{-x_i+\ln(\gamma)}-e^{-x_i})  \right) \\
    &= (\frac1\gamma - 1) \| \A x y \|_1 + \ln(\gamma)  +  \frac{\eps^2}{n e^B}(\frac{1}{\gamma}-1)\left(\sum_i e^{x_i}\right) + \frac{\eps^2}{n e^B}(\gamma-1)(\sum_i e^{-x_i})   \\
    &\leq (\frac1\gamma - 1) \|\A xy\|_1 + \ln(\gamma)  +0+ \frac{\eps^2}{n e^B}(\gamma-1)(\sum_i e^{-x_i})   \\
    &\leq (\frac1\gamma - 1) C' + \ln(\gamma) + (\gamma-1)(\ln(1/\mu)+\frac{4\eps^2}{e^B})  
\end{align*}
where for the last inequality we use $\norm{\A xy}_1 \geq C'$ for the first term and \cref{eq: sum exp bound} for the last term.
\end{proof}
An appropriate choice of $C'$ and $\gamma$ makes the bound in the above lemma non-positive. 
\begin{corollary} \label{cor: large 1 norm implies potential decrease}
Let $\eps \leq 1$ and $\mu \leq 1$, set $\gamma = 2$ and $C' = 2(\ln(2/\mu) + 4\eps^2/e^B)$. Then, if $\norm{\A x y}_1 \geq C'$ and $\tilde f(\vec x, \vec y) \leq \tilde f(\vec 0, \vec 0)$, we have $\tilde f(\vec x', \vec y) \leq \tilde f(\vec x, \vec y)$.
\end{corollary}

\subsection{Quantum box-constrained scaling}

Combining the above leads to a quantum algorithm for matrix scaling that is based on classical box-constrained newton methods. See \cref{alg:quantum box-constrained} for its formal definition. In \cref{thm:upper bound} we analyze its output. 
\begin{algorithm}[ht]
    \caption{Quantum box-constrained Newton method for matrix scaling}
    \label{alg:quantum box-constrained}
    \Input{Oracle access to $\mat A \in [0,1]^{n \times n}$ with $\norm{\mat A}_1 \leq 1$ and smallest non-zero entry $\mu > 0$, error $\eps > 0$, targets $\vec r, \vec c \in \R_{>0}^n$ with $\norm{\vec r}_1 = 1 = \norm{\vec c}_1$, diameter bound $B \geq 1$, classical $k$-oracle $\mathcal A$ for SDD matrices with non-negative off-diagonal entries}
    \Output{Vectors $\vec x, \vec y \in \R^n$ with $\norm{(\vec x, \vec y)}_\infty \leq B + \ln(4n + (n \ln(1/\mu)/\eps^2))$}
    set $T = \lceil 4 e^4 \max(k B + \ln (4n + (n \ln(1/\mu) / \eps^2)),1) \cdot \ln\left(\frac{\ln(1/\mu) + 2\eps^2/e^B}{\eps^2 / 2}\right)\rceil$\;
    set $C' = 2 \lceil \ln(2/\mu) + 8 \eps^2 / e^B \rceil$\; 
    set $\eps' = \lfloor \eps^2 / 8e^4 \max(k (B + \ln (4n + (n \ln(1/\mu) / \eps^2))), 1) \rfloor$\; 
    store $\vec x^{(0)}, \vec y^{(0)} = \vec 0 \in \R^n$ in QCRAM\;
    \For{$i = 0, \dotsc, T-1$}{
        compute $\mat H_m$, $\mat H_a$ s.t. $\mat H_m + \mat H_a \approx \hess \tilde f(\vec x^{(i)}, \vec y^{(i)})$ as in \cref{thm:approximating regularized potential Hessian} with $\delta_a = \eps' k^2 / 2e^2$\;
        compute $\vec b \approx \grad \tilde f(\vec x^{(i)}, \vec y^{(i)})$ as in \cref{lem:approxgrad} at $\vec x^{(i)}, \vec y^{(i)}$ with $\delta = \eps' / 3$\;
        compute $\Delta = \mathcal A(\tfrac{4e^2}{3k^2} \cdot (\mat H_m + \mat H_a), \tfrac{\vec b}{k})$\;
        compute $(\vec x^{(i+1)}, \vec y^{(i+1)}) = (\vec x^{(i)}, \vec y^{(i)}) + \frac{1}{k} \Delta$ and store in QCRAM\;\label{line: update Delta}
        set \texttt{flag = true}\;
        \While{\texttt{flag}}
        {
        Compute $C'/2$-additive approximation $\gamma$ of $\norm{\mat A(\vec x^{(i+1)}, \vec y^{(i+1)})}_1$\;
        \If{$\gamma \leq 3C'/2$}{set \texttt{flag = false}\;}
        \Else{
        update $\vec x^{(i+1)} \leftarrow \vec x^{(i+1)} - \ln(2) \vec 1$ in QCRAM\;
        }
        }
    }
    \Return{$(\vec x, \vec y) = (\vec x^{(T)}, \vec y^{(T)})$}\;
\end{algorithm}

\begin{theorem}
  \label{thm:upper bound}
  Let $\mat A \in [0,1]^{n \times n}$ with $m$ non-zero entries, $\vec r, \vec c \in \R_{>0}^n$ such that $\norm{\vec r}_1 = 1 = \norm{\vec c}$, and assume $\mat A$ is asymptotically $(\vec r, \vec c)$-scalable.
  Let $\eps > 0$, let $B \geq 1$, and assume there exist $(\vec x_\eps, \vec y_\eps)$ such that $\norm{(\vec x_\eps, \vec y_\eps)}_\infty \leq B$ and $f(\vec x_\eps, \vec y_\eps) - f^* \leq \eps^2$.
  Furthermore, let $\mathcal A$ be the $\bigO{\log(n)}$-oracle of \cref{thm: log n oracle}.
  Then \cref{alg:quantum box-constrained} with these parameters outputs, with probability $\geq 2/3$, vectors $\vec x, \vec y$ such that $f(\vec x, \vec y) - f^* \leq 6 \eps^2$ and runs in quantum time $\bigOt{B^2 \sqrt{mn} / \eps^2}$.
\end{theorem}
\begin{proof}
  In every iteration, the matrices $\mat H_m, \mat H_a$ and the vector $\vec b$ are such that they satisfy the requirements of \cref{thm: second order robust update}, hence
  \[
    \tilde f(\vec x^{(i+1)}, \vec y^{(i+1)}) - \tilde f^* \leq \left(1 - \frac{1}{4e^4 \max(k R_\infty,1)} \right) (\tilde f(\vec x^{(i)}, \vec y^{(i)}) - \tilde f^*) + \frac{e^2\delta_a}{k^2} + \frac{3\delta}{2}
  \]
  where $R_\infty \leq B + \ln (4n + (n \ln(1/\mu) / \eps^2)))$ is the $\ell_\infty$-radius of the sublevel set $\{(\vec x, \vec y) : \tilde f(\vec x, \vec y) \leq \tilde f(\vec 0, \vec 0) \}$ about $(\vec 0, \vec 0)$, whose upper bound follows from \cref{lem:basic regularized properties}.
  From here on, we write $M = 4 e^4 \max(k R_\infty, 1)$.
  The choice of $\delta_a$ and $\delta$ in the algorithm is such that $e^2 \delta_a /k^2 + 3 \delta / 2 \leq \frac{\eps^2}{2M}$, hence we can also bound the progress by
  \[
    \tilde f(\vec x^{(i+1)}, \vec y^{(i+1)}) - \tilde f^* \leq \left(1 - \frac{1}{M} \right) (\tilde f(\vec x^{(i)}, \vec y^{(i)}) - \tilde f^*) + \frac{\eps^2}{2M}.
  \]
  
  \Cref{cor: large 1 norm implies potential decrease} shows that if $\norm{\mat A(\vec x^{(i+1)}, \vec y^{(i+1)})}_1$ is larger than $C'$, then we can  shift $\vec x$ by $-\ln(2) \vec 1$, this halves $\norm{\mat A(\vec x^{(i+1)}, \vec y^{(i+1)})}_1$ and does not increase the regularized potential. Repeating this roughly $\log_2(\|\mat A(\vec x^{(i+1)}, \vec y^{(i+1)})\|_1/C')$ many times\footnote{Which is an almost constant number of times: in a single update of the box-constrained method, we take steps of size at most $1$ in $\ell_\infty$-norm, so individual entries can only grow by a factor $e^2$ in a single iteration, and the holds same for $\|\A xy\|_1$.} 
  reduces $\|\mat A(\vec x^{(i+1)}, \vec y^{(i+1)})\|_1$ to at most $C = 2C'$. Determining when to stop this process requires a procedure to distinguish between the cases $\norm{\mat A(\vec x^{(i+1)}, \vec y^{(i+1)})}_1 \leq C'$ and $\norm{\mat A(\vec x^{(i+1)}, \vec y^{(i+1)})}_1 \geq 2C'$ (if in between $C'$ and $2C'$ either continuing or stopping is fine). Such a procedure can be implemented by computing a $C'/2$-additive approximation of $\norm{\mat A(\vec x^{(i+1)}, \vec y^{(i+1)})}_1$, which can be done using $\bigOt{\sqrt{mn}\polylog(C'/\mu)}$ quantum queries, see (the proof of)~\cite[Lemma 4.6 (arXiv)]{qscalingICALP}.
  Therefore, throughout the algorithm we may assume that $\norm{\mat A(\vec x^{(i+1)}, \vec y^{(i+1)})}_1 \leq 2 C' = C$.
  
  It remains to show that $\tilde f(\vec x^{(T)}, \vec y^{(T)}) - \tilde f^* \leq \eps^2$ for our choice of $T$.
  Note that we have
  \begin{align*}
      \tilde f(\vec x^{(T)}, \vec y^{(T)}) - \tilde f^* & \leq (1 - \frac1M)^T (\tilde f(\vec 0, \vec 0) - \tilde f^*) + \sum_{i=0}^{T-1} (1 - \frac1M)^{T-i-1} \cdot \frac{\eps^2}{2M} \\
      & \leq (1 - \frac1M)^T (\tilde f(\vec 0, \vec 0) - \tilde f^*) + (1 - (1 - \frac1M)^T) \cdot \frac{\eps^2}{2} \\
      & \leq (1 - \frac1M)^T (f(\vec 0, \vec 0) - f^* + \frac{2\eps^2}{e^B}) + \frac{\eps^2}{2} \\
      & \leq (1 - \frac1M)^T (\ln(1/\mu) + \frac{2 \eps^2}{e^B}) + \frac{\eps^2}{2} \\
      & \leq \eps^2
  \end{align*}
  where in the third inequality we use \cref{lem:basic regularized properties}, and in the last inequality we use
  \begin{align*}
    T & = \left\lceil 4 e^4 \max(k B + \ln (4n + (n \ln(1/\mu) / \eps^2)),1) \cdot \ln\left(\frac{\ln(1/\mu) + 2\eps^2/e^B}{\eps^2 / 2}\right) \right\rceil \\
    & \geq \left\lceil M \cdot \ln\left(\frac{\ln(1/\mu) + 2\eps^2/e^B}{\eps^2 / 2}\right) \right\rceil \\
    & \geq \frac{1}{\ln(1 - \frac1M)} \cdot \ln\left(\frac{\eps^2/2}{\ln(1/\mu) + \frac{2 \eps^2}{e^B}}\right). 
  \end{align*} 
  This implies that
  \begin{align*}
      f(\vec x^{(T)}, \vec y^{(T)}) - f^* \leq \tilde f(\vec x^{(T)}, \vec y^{(T)}) - \tilde f^* + 5 \eps^2 \leq 6 \eps^2,
  \end{align*}
  where we crucially use the last point of \cref{lem:basic regularized properties} and the assumption that there exist $(\vec x_\eps, \vec y_\eps)$ with $\norm{(\vec x_\eps, \vec y_\eps)}_\infty \leq B$ which $\eps^2$-minimize $f$.
  
  Finally we bound the time complexity of \cref{alg:quantum box-constrained}. For each of the quoted results, we use the choice $C = 2 C' = \bigOt{\ln(n)+\eps^2}$. In each of the $T$ iterations we compute: 
  \begin{enumerate}
      \item approximations $\mat H_m$, $\mat H_a$ of $\nabla^2 \tilde f(\vec x^{(i)}, \vec x^{(i)})$ in time $\bigOt{\sqrt{mn}\polylog(1/\eps)}$ (using that $C$, $1/\mu$ are at most $\poly(n)$), 
      \item an $\eps'/3$-$\ell_1$-approximation of $\nabla \tilde f (\vec x^{(i)},\vec y^{(i)})$ in time $\bigOt{\sqrt{mn}/\eps'} = \bigOt{B\sqrt{mn}/\eps^2}$,
      \item an update $\Delta$ in time $\bigOt{n}$ 
      using one call to the $k = \bigO{\log(n)}$-oracle on SDD-matrices with $\bigOt{n}$ non-zero entries from \cref{thm: log n oracle},
      \item at most $\bigO{1}$ many times (using the fact that in \cref{line: update Delta} the $1$-norm changes by at most a constant factor since $\|\frac1k\Delta\|_\infty \leq 1$)       an $\bigO{\ln(1/\mu)+\eps^2}$-additive approximation of $\|\mat A({\vec x^{(i)}},{\vec y^{(i)}})\|_1$ in time $\bigOt{\sqrt{mn}}$.
    \end{enumerate}
  Note that the second contribution dominates the others, resulting in an overall time complexity of $\bigOt{B^2\sqrt{mn}/\eps^2}$.
\end{proof}

The above theorem assumes that a bound $B$ on the $\ell_\infty$-norm of an $\eps^2$-minimizer of $f$ is known.
For the purpose of matrix scaling, one can circumvent this assumption by running the algorithm for successive powers of $2$ (i.e., $B=1$, $B=2$, $B=4$,$\ldots$) and testing after each run whether the output provides an $\eps$-scaling or not.
Verifying whether given $\vec x, \vec y$ provide an $\eps$-scaling of $\mat A$ can be done in time $\bigOt{\sqrt{mn}/\eps^2}$.
Note that this gives an algorithm for $\eps$-scaling whose complexity depends on a diameter bound for $\eps^2$-minimizers of $f$, rather than a diameter bound for $\eps$-scaling vectors.
Furthermore, such an approach does not work for the task of finding an $\eps^2$-minimizer of $f$, as we do not know how to test this property efficiently. 

\begin{corollary} \label{cor: non-positive}
  For asymptotically-scalable matrices $\mat A \in \R^{n \times n}_{\geq 0}$ with $m$ non-zero entries, one can find  $O(\eps)$-$\ell_1$-scaling vectors $(\vec x, \vec y)$ of $\mat A$ to target marginals $\vec r, \vec c \in \R^n_{>0}$ with $\norm{\vec r}_1 = 1 = \norm{\vec c}_1$ in time $\bigOt{R_\infty^2 \sqrt{mn}/\eps^2}$, where $R_\infty$ is such that there exists an $\eps^2$-approximate minimizer $(\vec x_\eps, \vec y_\eps)$ of $f$ with
  \[R_\infty = \|(\vec x_\eps, \vec y_\eps)\|_\infty + \ln(4n+ (n \ln(1/\mu)/\eps^2)).
  \]
\end{corollary}
For the general case mentioned above, we do not have good (i.e., polylogarithmic) bounds on the parameter $R_\infty$. We do have such bounds when $\mat A$ is entrywise positive: it is well-known (and easy to show\footnote{From the inequality $A_{ij} e^{x_i + y_j} \leq 1/n$ one gets the upper bounds $x_i + y_j \leq \ln(1/n\mu)$ for every $i, j$. To obtain a variation norm bound for $\vec x$ and $\vec y$, note that for every fixed $i$, there is at least one $j_i$ such that $A_{ij_i} e^{x_i + y_{j_i}} \geq 1/n^2$ (because the row sums are $1/n$). Therefore $x_i + y_{j_i} \geq \ln(1/n^2 \nu)$ where $\nu$ is the largest entry of $\mat A$, and $x_{i'} - x_i = (x_{i'} + y_{j_i}) - (x_i + y_{j_i}) \leq \ln(1/n\mu) - \ln(1/n^2 \nu) = \ln(n \nu / \mu)$ for every $i, i'$. This is an upper bound on the variation norm of $\vec x$, and one can derive the same bound for that of $\vec y$. By translating $\vec x, \vec y$ by appropriate multiples of the all-ones vector we can assume $x_1 = 0$. Then the variation-norm bound also bounds the $\ell_\infty$-norm of $\vec x$. To then get an $\ell_\infty$-bound on $\vec y$, recall that for at least one $j$, one has $x_1 + y_j = y_j \geq \ln(1/n^2 \nu) \geq 0$ and still $y_j = x_1 +y_j \leq \ln(1/n\mu)$, so $\norm{\vec y}_\infty \leq \ln(1/n\mu) + \ln(n \nu/\mu) = \ln(\nu/\mu^2)$.})
that such an $\mat A$ can be exactly scaled to uniform marginals with scaling vectors $(\vec x, \vec y)$ such that $\norm{(\vec x, \vec y)}_\infty = O(\log(\norm{\mat A}_1 / \mu))$
(cf.~\cite[Lem.~1]{KALANTARI199687}, \cite[Lem.~4.11]{cmtv17}).
In particular, this implies that there exists a minimizer $(\vec x^*,\vec y^*)$ of $f$ with $\norm{(\vec x^*,\vec y^*)}_\infty = O(\log(\norm{\mat A}_1/\mu)) = \bigOt{1}$ and therefore we have the following corollary.

\begin{corollary} \label{cor: positive}
For entrywise-positive matrices $\mat A$, one can find an $\eps$-$\ell_1$-scaling of $\mat A$ to uniform marginals in time $\bigOt{n^{1.5} / \eps^2}$.
\end{corollary}

\paragraph{Optimality of the choice of parameters.}

Let $z_i = \tilde f(\vec x_i, \vec y_i) - \tilde f(\vec x^*,\vec y^*)$ for each iteration $i$. Then the $z_i$ satisfy constraints of the following form: 
\[
z_{i+1} \leq (1-\gamma) z_i + \delta_i,
\]
where $\gamma = \frac{1}{4e^4 \max(k R_\infty,1)}$ and $\delta_i$ is a parameter that determines the accuracy with which we approximate the gradient and Hessian in each iteration.
Above we used the choice $\delta_i = \frac{e^2\delta_a}{k^2} + \frac{3\delta}{2}$, independent of $i$. Since $1/\delta_i$ dominates the complexity of each iteration, a natural question is whether one can obtain a better overall complexity by letting $\delta_i$ depend on $i$. In the following lemma we show this is not the case.
\begin{lemma}
    Let $z_0 > 0$, $\eps > 0$ and $0 < \gamma \leq 1/2$ be given.
    Then, for any $N \geq 1$ and any choice of sequence of $\delta_0, \dotsc, \delta_{N-1} > 0$ such that the sequence defined by
    \begin{equation*}
        z_{i+1} = \left(1 - \gamma\right) z_i + \delta_i, \quad 0 \leq i \leq N-1
    \end{equation*}
    satisfies $z_N \leq \eps$, one must have
    \begin{equation*}
        \sum_{i=0}^{N-1} \frac1{\delta_i} \geq \frac{1}{\gamma^2 \eps} \cdot (1 - \sqrt{\eps/z_0})^2.
    \end{equation*}
\end{lemma}
\begin{proof}
  Observe that we have the explicit expression
  \[
    z_N = (1 - \gamma)^N z_0 + \sum_{i=0}^{N-1} (1 - \gamma)^{N-i-1} \delta_i.
  \]
  As every term in this sum is positive, we must have $(1 - \gamma)^N z_0 < \eps$ if $z_N \leq \eps$ (where we have strict inequality since $N \geq 1$ and therefore the sum is not empty).
  Now fix $N$ such that $(1 - \gamma)^N z_0 < \eps$, and define the Lagrangian $L(\delta_0, \dotsc, \delta_{N-1}; \lambda)$ by
  \begin{equation} \notag
      L(\delta_0, \dotsc, \delta_{N-1}; \lambda) = \sum_{i=0}^{N-1} \frac{1}{\delta_i} + \lambda \left((1 - \gamma)^N z_0 + \sum_{i=0}^{N-1} (1 - \gamma)^{N-i-1} \delta_i - \eps\right).
  \end{equation}
  Observe that the Lagrangian is convex in the $\delta_i$ and that the constraint $z_N \leq \eps$ is linear in the $\delta_i$, and can be made strict for a very small choice of $\delta_i$. In other words, the Karush--Kuhn--Tucker conditions are satisfied, so that $\sum_{i=0}^{N-1} 1/\delta_i$ is minimized subject to the constraint $z_N \leq \eps$ if and only if $\grad L(\delta_0, \dotsc, \delta_{N-1}; \lambda) = 0$ for some $\lambda \geq 0$. 
  This gradient vanishes if and only if
  \begin{equation} \notag
      - \frac{1}{\delta_i^2} + \lambda \cdot \left(1 - \gamma\right)^{N-i-1} = 0, \quad 0 \leq i \leq N-1, \quad \eps = (1 - \gamma)^N z_0 + \sum_{i=0}^{N-1} (1 - \gamma)^{N-i-1} \delta_i.
  \end{equation}
  For fixed $\lambda > 0$ this means that the optimal choice of $\delta_i$ is
  \begin{equation} \notag
      \delta_i = \sqrt{\frac{1}{\lambda (1 - \gamma)^{N-i-1}}} = c_\lambda \sqrt{1 - \gamma}^{-N+i+1}
  \end{equation}
  where $c_\lambda := \sqrt{1/\lambda}$. The constraint on $z_N$ then gives
  \begin{equation} \notag
      \eps - (1 - \gamma)^N z_0 = c_\lambda \sum_{i=0}^{N-1} \sqrt{1 - \gamma}^{N-i-1} = c_\lambda \cdot \frac{1 - \sqrt{1 - \gamma}^{N}}{1 - \sqrt{1 - \gamma}},
  \end{equation}
  leading to an associated cost of 
  \begin{equation} \notag
      \sum_{i=0}^{N-1} \frac{1}{\delta_i} = \frac{1}{c_\lambda} \cdot \frac{1 - \sqrt{1 - \gamma}^N}{1 - \sqrt{1 - \gamma}} = \left(\frac{1 - \sqrt{1 - \gamma}^N}{1 - \sqrt{1 - \gamma}}\right)^2 \cdot \frac{1}{\eps - (1 - \gamma)^N z_0}.
  \end{equation}
  As $\gamma \leq 1$ we have $1 - \sqrt{1 - \gamma} \leq \gamma$, and because $(1 - \gamma)^N z_0 < \eps$, we have 
  \begin{equation*}
      1 - \sqrt{1 - \gamma}^N > 1 - \sqrt{\frac{\eps}{z_0}}
  \end{equation*}
  and the cost satisfies
  \begin{equation*}
      \sum_{i=0}^{N-1} \frac{1}{\delta_i} \geq \left(\frac{1 - \sqrt{\eps/z_0}}{\gamma}\right)^2 \cdot \frac{1}{\eps - (1 - \gamma)^N z_0} = \frac{1}{\gamma^2(\eps - (1 - \gamma)^N z_0)} \cdot (1 - \sqrt{\eps/z_0})^2 \geq \frac{1}{\gamma^2 \eps} \cdot (1 - \sqrt{\eps/z_0})^2.
  \end{equation*}
\end{proof}

\noindent\textbf{Acknowledgments}\; We thank Joran van Apeldoorn, Michael Walter and Ronald de Wolf for interesting and helpful discussions, and the latter two for giving feedback on a first version of this paper. Moreover, we thank Ronald de Wolf for pointing us to \cite{lee&roland:qsdpt}, which allows for an exponentially small success probability in \cref{thm:basic lb}.
    
\bibliographystyle{alpha}
\bibliography{references}

\end{document}